\documentclass[a4paper,UKenglish,cleveref, autoref, thm-restate]{lipics-v2021}
\usepackage{todonotes}
\usepackage{float}

\newcommand{\BiFore}{\texttt{foreL}}
\newcommand{\BiBack}{\texttt{backL}}
\newcommand{\Fore}{\texttt{fore}}
\newcommand{\Back}{\texttt{back}}
\newcommand{\ForeAll}{\texttt{forePair}}

\newcommand{\col}{\texttt{col}}
\newcommand{\subIB}{\texttt{SubIB}}
\newcommand{\subIF}{\texttt{SubIF}}
\newcommand{\inter}{\texttt{intervals}}
\DeclareMathOperator{\rank}{\texttt{rank}}
\DeclareMathOperator{\select}{\texttt{select}}
\newcommand{\valB}{\texttt{ValBack}}
\newcommand{\valF}{\texttt{ValFore}}
\newcommand{\width}{\texttt{w}}
\newcommand{\height}{\texttt{h}}
\newcommand{\word}{w}
\newcommand{\occ}{\texttt{occ}}
\newcommand{\paInterval}{\texttt{pa-interval}}
\newcommand{\sPA}{\texttt{sPA}}
\newcommand{\sValS}{\texttt{sVal}}
\newcommand{\sCountS}{\texttt{sCount}}
\newcommand{\PA}{\texttt{PA}}
\newcommand{\PBWT}{\texttt{PBWT}}

\newcommand{\newR}{\tilde{r}}

\title{Optimal-Time Mapping in Run-Length Compressed PBWT} %TODO Please add

\titlerunning{Optimal-Time Mapping in Run-Length Compressed PBWT} %TODO optional, please use if title is longer than one line

\author{Paola Bonizzoni}{Department of Computer Science, University of Milano-Bicocca, Italy}{paola.bonizzoni@unimib.it}{https://orcid.org/0000-0001-7289-4988}{}

\author{Davide Cozzi}{Department of Computer Science, University of Milano-Bicocca, Italy}{d.cozzi@campus.unimib.it}{https://orcid.org/0000-0003-2439-0608}{}

\author{Younan Gao}{Department of Computer Science, University of Milano-Bicocca, Italy}{younan.gao@unimib.it}{https://orcid.org/0000-0003-4984-2551}{}

\authorrunning{Bonizzoni et al.} %TODO mandatory. First: Use abbreviated first/middle names. Second (only in severe cases): Use first author plus 'et al.'

\Copyright{Paola Bonizzoni, Davide Cozzi, and Younan Gao} %TODO mandatory, please use full first names. LIPIcs license is "CC-BY";  http://creativecommons.org/licenses/by/3.0/


 %TODO mandatory; please add comma-separated list of keywords

\category{} %optional, e.g. invited paper

\relatedversion{} %optional, e.g. full version hosted on arXiv, HAL, or other respository/website
\acknowledgements{The authors gratefully acknowledge Travis Gagie for bringing prefix searches to their attention.}%optional

\nolinenumbers %uncomment to disable line numbering

\EventEditors{John Q. Open and Joan R. Access}
\EventNoEds{2}
\EventLongTitle{42nd Conference on Very Important Topics (CVIT 2016)}
\EventShortTitle{CVIT 2016}
\EventAcronym{CVIT}
\EventYear{2016}
\EventDate{December 24--27, 2016}
\EventLocation{Little Whinging, United Kingdom}
\EventLogo{}
\SeriesVolume{42}
\ArticleNo{23}
\begin{document}

\maketitle

\begin{abstract}
The Positional Burrows--Wheeler Transform (PBWT) is a data structure designed for efficiently representing and querying large collections of sequences, such as haplotype panels in genomics. 
Forward and backward stepping operations---analogues to LF- and FL-mapping in the traditional BWT---are fundamental to the PBWT, underpinning many algorithms based on the PBWT for haplotype matching and related analyses. 
Although the run-length encoded variant of the PBWT (also known as the $\mu$-PBWT) achieves $O(\newR)$-word space usage, where $\newR$ is the total number of runs, no data structure supporting both forward and backward stepping in constant time within this space bound was previously known. 
In this paper, we consider the multi-allelic PBWT that is extended from its original binary form to a general ordered alphabet $\{0, \dots, \sigma-1\}$.
We first establish bounds on the size $\newR$ and then introduce a new $O(\newR)$-word data structure built over a list of haplotypes $\{S_1, \dots, S_\height\}$, each of length $\width$, that supports constant-time forward and backward stepping. 

We further revisit two key applications---haplotype retrieval and prefix search---leveraging our efficient forward stepping technique.
Specifically, we design an $O(\newR)$-word space data structure that supports haplotype retrieval in $O(\log \log_{\word} h + \width)$ time.
For prefix search, we present an $O(\height + \newR)$-word data structure that answers queries in $O(m' \log\log_{\word} \sigma + \occ)$ time, where $m'$ denotes the length of the longest common prefix returned and $\occ$ denotes the number of haplotypes prefixed the longest prefix.
\end{abstract}

\newpage
\section{Introduction}
\label{sect-intro}

\textbf{Background and motivation.} The \emph{Positional Burrows--Wheeler Transform (PBWT)}~\cite{Durbin2014-dd} is a data structure designed for efficiently representing and querying large collections of sequences, such as haplotype panels in genomics. 
%Originally proposed by Durbin~\cite{Durbin2014-dd}, the PBWT  stores a set of $\height$ haplotypes across $\width$ variant sites in a $\height \times \width$ binary matrix by reporting at each column $j$   the permutations of the $\height$ rows   arranged in co-lexicographic order according to the prefixes of the haplotypes up to column $j-1$. 
Originally proposed by Durbin~\cite{Durbin2014-dd}, the PBWT stores a set of $\height$ haplotypes across $\width$ variant sites in a $\height \times \width$ binary matrix, where the rows at each column $j$ are arranged in co-lexicographic order according to the prefixes of the haplotypes up to column $j-1$.
%This ordering enables efficient querying between a given haplotype and the panel \textcolor{red}{and} the identification of set-maximal exact matches (SMEMs).
This ordering facilitates efficient querying between a given haplotype and the panel, enabling the identification of set-maximal exact matches (SMEMs).

Two fundamental operations are defined on the PBWT: \emph{forward stepping} and \emph{backward stepping}. Forward stepping, denoted by $\Fore[i][j]$, maps the position of a haplotype in the permutation of the PBWT at site $j$ to its position in the permutation at site $j{+}1$, while backward stepping, denoted by $\Back[i][j]$, performs the inverse mapping, tracing a haplotype’s position at site $j{+}1$ back to its position at site $j$. 
These operations play an essential role in many PBWT-based algorithms for haplotype matching and analysis.

Although the PBWT enables efficient computations, its memory usage grows rapidly with large haplotype datasets, posing a challenge for population-scale cohorts like the UK Biobank \cite{halldorsson2022sequences}.
To address this limitation, the \emph{$\mu$-PBWT}~\cite{Cozzi2023-iv}, also discussed in \cite{Bonizzoni2023-cj}, was introduced as a compressed variant of the PBWT that leverages \emph{run-length encoding (RLE)} to reduce space usage. A \emph{run} in a sequence is defined as a maximal contiguous block of identical symbols, and the $\mu$-PBWT$'$s storage requirement is $O(\newR)$ words, where $\newR$ denotes the total number of runs in the corresponding PBWT across all $\width$ sites. %This run-length compression allows the $\mu$-PBWT$'$s memory footprint to be only a fraction of that of the original PBWT, making it feasible to index and query massive haplotype panels such as those in the UK Biobank.
This run-length compression reduces the $\mu$-PBWT's memory usage to only a fraction of that of the original PBWT, making it feasible to index and query massive haplotype panels such as those in the UK Biobank \cite{halldorsson2022sequences}.

Beyond space efficiency, the $\mu$-PBWT's structure has proven particularly useful for downstream applications such as computing matching statistics~\cite{Cozzi2023-iv, Bonizzoni2024-dj}, Minimal Positional Substring Covers (MPSC)~\cite{Bonizzoni2024-dj}, and SMEMs~\cite{Cozzi2023-iv, Bonizzoni2024-dj}. Each of these applications relies on repeated forward and backward stepping operations. However, within the $O(\newR)$-space bound of the $\mu$-PBWT, no data structure supporting both operations in constant time was known prior to our work. Achieving faster stepping operations would directly improve the performance of all these applications and further enhance the efficiency of run-length compressed PBWT indexing for large-scale genomic analyzes.

\smallskip
\noindent
\textbf{Related work.} 
Durbin~\cite{Durbin2014-dd} noted that, due to the co-lexicographic ordering of the PBWT, the {forward} (resp, {backward}) stepping operation is the natural analogue of the LF- (resp, FL-) mapping in the classical Burrows--Wheeler Transform (BWT). 
Gagie et al. \cite[Lemma 2.1]{GagieNP20} presented a data structure requiring $O(r)$ words of space, built over a text of length $n$, that supports LF- and FL-mappings on the BWT in $O(\log \log_{\word} (n/r))$ time, where $r$ is the number of runs in the BWT of the text and $\word$ is the number of bits in a machine word.
%
% Specifically, let $C_j$ for $1 \le j \le \width$ denote the number of occurrences of bit~$0$ at site~$j$ of the \textcolor{red}{(bi-allelic)} PBWT. 
% Given the position $i$ in the permutation of the PBWT at site~$j$, let $x$ denote the number of occurrences of the bit stored at position~$i$ among the first $i$ entries at site~$j$. 
% The forward stepping operation, denoted by $\Fore[i][j]$, computes the corresponding position at site~$j{+}1$ as $x$ if the bit at position~$i$ is~$0$, and as $C_j + x$ otherwise. 
% The backward stepping operation, denoted by $\Back[i][j]$, can be implemented symmetrically. 
By applying this data structure to each column of the PBWT, we can achieve, for every site~$j$, an $O(r_j)$-word data structure that supports both forward and backward stepping queries in $O(\log \log_{\word} (\height / r_j))$ time, where $r_j$ denotes the number of runs at site~$j$.

Nishimoto and Tabei~\cite{NishimotoT21} recently proposed the \emph{move structure} to accelerate LF-mapping operations on the BWT. 
Their approach first divides the $r$ runs in the BWT into at most $2r$ \emph{sub-runs}, and then constructs an $O(r)$-word data structure over these sub-runs. 
Given a position $i$ in the BWT and the index of the sub-run containing $i$, the structure can compute LF$(i)$ and determine the index of the sub-run containing LF$(i)$ in constant time.

Typically, in applications of the PBWT, forward and backward stepping are performed iteratively.
For example, to traverse a haplotype in the PBWT from site~$j$ to site~$j'$, the forward stepping procedure is invoked sequentially $j' - j$ times.
However, efficient (e.g., constant-time) forward and backward stepping in the run-length encoded PBWT cannot be achieved simply by applying the move structure independently at each site.
Indeed, the move structure only returns the index of the sub-run at site~$j$ that contains $\Fore[i][j]$, whereas continuing the forward step from site~$j+1$ to site~$j+2$ requires the index of the sub-run at site~$j+1$ that contains $\Fore[i][j]$.
See Figure \ref{fig:pbwt-pa-exp} (in Appendix \ref{app-pbwt-pa-exp}) for an example.
Hence, new methods are needed to accelerate $\Fore$ and $\Back$ queries while maintaining the run-length encoded PBWT.

% \begin{figure}[!t]
%   \centering
%   \begin{subfigure}[b]{0.3\textwidth}
%     \centering
%     \includegraphics[width=\textwidth]{figs/hap-exp}
%     \caption{The Input haplotypes}
%     \label{fig:hap-exp}
%   \end{subfigure}
%   \hfill
%   \begin{subfigure}[b]{0.3\textwidth}
%     \centering
%     \includegraphics[width=\textwidth]{./figs/pbwt-exp}
%     \caption{The PBWT with runs}
%     \label{fig:pbwt-exp}
%   \end{subfigure}
%   \hfill
%   \begin{subfigure}[b]{0.3\textwidth}
%     \centering
%     \includegraphics[width=\textwidth]{./figs/pa-exp}
%     \caption{PA}
%     \label{fig:pa-exp}
%   \end{subfigure}
%   \caption{Example of $\PBWT$ and $\PA$ built for bi-allelic haplotypes. The operation $\Fore[5][4]$ returns $2$, and the position 2 is in the first run of the fifth column of the PBWT.}
%   \label{fig:pbwt-pa-exp}
% \end{figure}

Two key applications of the $\Back$ and $\Fore$ queries on the PBWT are \emph{haplotype retrieval} and \emph{prefix search}~\cite{Gagie2022-js}.
In haplotype retrieval, the goal is to extract any haplotype from the PBWT.
In prefix search, the objective is to find the longest common prefix $P[1..m']$ between any haplotype and a query pattern $P[1..m]$, and to list the indices of all haplotypes prefixed by $P[1..m']$.
Using the data structure from~\cite[Lemma~2.1]{GagieNP20}, built over each column of the PBWT, haplotype retrieval can be performed in $O(\width \log \log_{\word} h)$ time by invoking forward stepping iteratively $\width$ times, with an overall space usage of $O(\newR)$ words.
Gagie et al.~\cite{Gagie2022-js} present a textbook solution based on the PBWT that uses $O(\height \cdot \width)$ words of space and supports each prefix query in $O(m \log \height + \occ)$ time, where $\occ$ denotes the number of reported indices.
They also provide various time-space tradeoffs in the same work.
Nishimoto and Tabei~\cite{NishimotoT21} propose a data structure for prefix search that combines a BWT built over all haplotypes, a \emph{compact trie}~\cite{Morrison1968}, the move data structure, and at most $2\height - 1$ marked positions. The structure supports computing the longest prefix $P[1..m']$ in $O(m')$ time, using $O(r^{*} + \height)$ words of space, where $r^{*}$ denotes the number of runs in the BWT of the concatenated sequence.

\smallskip
\noindent
\textbf{Our  results.} 
We consider the multi-allelic PBWT as described in \cite{naseri2019multi}. We propose a ``move data structure''$-$style solution specifically tailored for the PBWT. Our approach introduces a simple yet effective algorithm that partitions the $\newR$ runs of the PBWT into at most $2\newR$ sub-runs across all $\width$ columns of the PBWT (Section \ref{sect-three-overlap-normal}).
We then design an $O(\newR)$-word data structure over the sub-runs (Theorem \ref{theorem-PBWT-r}) that supports $O(1)$-time computation of both forward and backward stepping operations, allowing to be invoked iteratively (Section \ref{sect-sub-run-back}). %More importantly, our solution allows these stepping queries to be invoked iteratively. 
 
As a first application, we present a PBWT-based solution for prefix searches. Specifically, assuming all haplotypes have the same length, we design an $O(\newR + \height)$-word data structure that finds the longest common prefix $P[1..m']$ between any haplotype and a query pattern $P[1..m]$ in $O(m' \log \log_{\word} \sigma + \occ)$ time.  
When the haplotypes are sorted lexicographically, we can further reduce the space complexity from $O(\newR + \height)$ to $O(\newR)$.
Note that the problem setting considered by Nishimoto and Tabei~\cite{NishimotoT21} allows haplotypes of varying lengths. We also show that our PBWT-based solution can be adapted to this more general case.  
While our query time is not as efficient as that of Nishimoto and Tabei~\cite{NishimotoT21}, our approach has the advantage of not requiring a compact trie.
As a second application, we design an $O(\newR)$-word data structure to represent the haplotypes, supporting retrieval of any haplotype $S_i$ ($1 \leq i \leq \height$) in $O(\log \log_{\word} \height + \width)$ time (Theorem~\ref{theorem-PBWT-extract}).

% Before introducing our ``move''-like $O(\newR)$-word data  structure, we complement the results with lower bounds on the $O(\newR)$ in order to highlight connections between repetitiveness in the input data set $S$ and the values of measure $\newR$. This could be of interest for further future investigations of $\newR$ related to clustering of identical haplotype sequences.
Before introducing our ``move''-like $O(\newR)$-word data structure, we complement our results with lower and upper bounds on the measure $\newR$.
We establish a connection between $\newR$ and the number $\height''$ of adjacent haplotype pairs $(S_i, S_{i+1})$ such that $S_i \ne S_{i+1}$.
Specifically, we show that $\newR \ge \height'' + 1$ and $\newR\le \width(\height'' + 1)$.

%the number $r_j$ of runs in each column $j$ of the PBWT is at most $(\height'' + 1)$.
%These bounds reveal the relationship between the repetitiveness of the input data set $S$ and the value of $\newR$, suggesting that clustering identical haplotype sequences prior to constructing the PBWT may lead to more space-efficient representations.

% Before introducing our ``move''-like $O(\newR)$-word data structure, we complement our results with lower and upper bounds on the size $\newR$.
% We establish the connection between $\newR$ is relevant to the number $\height''$ of the pairs of $S_i$ and $S_{i+1}$ such that $S_i\ne S_{i+1}$.
% Specifically, we show that $\newR\ge \height''+1$ and $\newR \le \width(\height''+1)$.
% These bounds highlight connections between repetitiveness in the input data set $S$ and the values of measure $\newR$. It indicates that clustering identical haplotype sequences together before constructing the PBWT might result in more efficient space.

%This could be of interest for further future investigations of $\newR$ related to clustering of identical haplotype sequences.

\smallskip
\noindent
\textbf{Paper organization.} In Section~\ref{sect-prelim}, we introduce the notation and preliminary results used throughout the paper.
In Section \ref{sect-bounds-newR}, we establish lower and upper bounds on the size $\newR$.
Section~\ref{sect-three-overlap-normal} defines the \emph{three-overlap} constraint over two lists of intervals and presents an algorithm that divides intervals in one list into sub-intervals satisfying this constraint.  
Applying this algorithm, we describe in Section~\ref{sect-sub-run-back} how to divide PBWT runs into sub-runs for $\Back$ and $\Fore$ queries.  
In Section~\ref{sect-ds-fore-back}, we design data structures over these sub-runs and develop the corresponding algorithms for $\Back$ and $\Fore$ queries.  
Section~\ref{sect-app} discusses two applications of $\Back$ and $\Fore$ queries: haplotype retrieval and prefix searches.  
Finally, Section~\ref{sect-conclusion} concludes the paper and outlines directions for future work.

\section{Preliminaries}
\label{sect-prelim}

All results in this paper are presented under the word RAM (random-access machine) model. 
We evaluate the space cost of data structures in words.
Each word is of $\word$ bits.

\smallskip
\noindent
\textbf{Notations.}
We denote by $[i,j]$ the interval of integers ${i, i+1, \dots, j}$, and define $[i,j] = \emptyset$ if $j < i$.
Given an interval $\mathsf{I} = [i,j]$, we denote its left and right endpoints by $\mathsf{I}.b$ and $\mathsf{I}.e$, respectively, so that $\mathsf{I}.b = i$ and $\mathsf{I}.e = j$.
For any matrix $A$, we denote by $\col_j(A)$ its $j$-th column and by $\col_j(A)[i]$ the entry $A[i][j]$.
Let $M$ be the matrix that stores $S_1 , S_2, \dots, S_\height$ in rows $1, 2, \dots, \height$. 
We note that the input haplotypes might not be pairwise distinct.
%We assume that all haplotypes are distinct. Hence, it follows that $\sigma^w\ge \height$.

Consider an ordered alphabet $\{0, \dots, \sigma-1\}$ with $0 < 1 < \dots < \sigma-1$. 
A \emph{string} $\alpha$ over this alphabet is a finite sequence of symbols from $\{0, \dots, \sigma-1\}$, that is, $\alpha = \alpha[1]\alpha[2]\dots\alpha[|\alpha|]$, where $|\alpha|$ denotes the length of $\alpha$, and $\alpha[i] \in \{0, \dots, \sigma-1\}$ for all $1 \le i \le |\alpha|$. 
The empty string is denoted by $\varepsilon$. 
For indices $1 \le i \le j \le |\alpha|$, we denote by $\alpha[i..j]$ the \emph{substring} of $\alpha$ spanning positions $i$ through $j$ (and define $\alpha[i..j] = \varepsilon$ if $i > j$). 
The string $\alpha[1..i]$ is referred to as the $i$-th \emph{prefix} of $\alpha$, and $\alpha[i..|\alpha|]$ as the $i$-th \emph{suffix}, for $1 \le i \le |\alpha|$. 
A \emph{proper prefix} (respectively, \emph{proper suffix}) of $\alpha$ is a prefix (respectively, suffix) $\beta$ such that $\beta \ne \alpha$.

Given two strings $\alpha$ and $\beta$ over the alphabet set $\{0,\dots, \sigma-1\}$, we say that $\alpha$ is \emph{lexicographically smaller} than $\beta$, denoted by $\alpha\prec \beta$, if and only if one of the following holds: i) There exists an index $k$ such that $\alpha[i] = \beta[i]$ for all $i < k$, and $\alpha[k] < \beta[k]$; ii) or $\alpha$ is a proper prefix of $\beta$.
We say that $\alpha$ is \emph{co-lexicographically smaller} than $\beta$, denoted by $\alpha \prec_{colex} \beta$, if and only if one of the following holds: i) There exists an index $k$ such that $\alpha[|\alpha|-i+1]=\beta[|\beta|-i+1]$ for all $1\le i <k$, and $\alpha[|\alpha|-k+1]<\beta[|\beta|-k+1]$; ii) or $\alpha$ is a proper suffix of $\beta$.

\smallskip
\noindent
\textbf{Predecessor queries.} Given a sorted list $S$ of integers, a predecessor query takes an integer $x$ as input and returns $\max \{ y \in S \mid y \leq x \}$ and its rightmost position in $S$.

\begin{lemma}\cite[Theorem A.1]{BelazzouguiN15}\label{lem-pred}
    Given an increasingly sorted list of $n'$ integers, drawn from the universe $\{0, \dots, \sigma-1\}$, there is a data structure that occupies $O(n' \log \sigma)$ bits of space and answers a predecessor query in $O(\log \log_{\word} \sigma)$ time.
\end{lemma}

\noindent
\textbf{Rank and select queries.}  
Given a sequence $A[1..n']$ over an alphabet $\{0, \dots, \sigma-1\}$, the operation $\rank_c(A, j)$ returns the number of occurrences of $c$ in $A[1..j]$, for $c \in \{0, \dots, \sigma-1\}$.  
The operation $\select_c(A, j)$ returns the position of the $j$-th occurrence of $c$ in $A$, for $1 \le j \le \rank_c(A, n')$, and returns $n'+1$ if $j > \rank_c(A, n')$.

\begin{lemma}\cite{BelazzouguiN15}\label{lem-rank-select}
There exists a data structure of $O(n')$ words built on $A[1..n']$ that supports $\rank$ queries in $O(\log \log_{\word} \sigma)$ time and $\select$ queries in $O(1)$ time.
\end{lemma}

\smallskip
\noindent
\textbf{Positional Burrows--Wheeler Transform (PBWT).}
Closely related to the PBWT is the \emph{Prefix Array (PA)}, a matrix that records, for each column, the permutation of haplotype indices induced by the PBWT.
Formally,   the Prefix Array $\PA$ built for the matrix $M$ is an $\height\times \width$ matrix, in which $\col_1(\PA)$ is simply the list $1, 2, \dots, \height$, and $\col_j(\PA)$, for $j>1$, stores the permutation of the set $\{1, \dots, \height\}$ induced by the co-lexicographic ordering of prefixes of $\{S_1, \dots, S_\height\}$ up to column $j-1$, that is, $\col_j(\PA)[i] = k$ if and only if $S_k[1..j-1]$ is ranked $i$ in the co-lexicographic order of $S_1[1..j-1], \dots, S_{\height}[1..j-1]$. 

Let $\PBWT$ be the matrix representing the positional BWT of $M$. Then $\PBWT$ is also an $\height \times \width$ matrix in which $\col_j(\PBWT)[i] = \col_j(M)[\col_j(\PA)[i]]$ for all $i \in [1..\height]$ and $j \in [1..\width]$.
We refer to a maximal substring of identical characters in $\col_j(\PBWT)$ as a \emph{run}.
Throughout the paper, let $r_j$ denote the number of runs for $\col_{j}(\PBWT)$. We define $\newR$ as $\sum_{1\le j\le \width} r_j$.

% Let us recall that given a column $j$ and row $i$ of the PBWT,   the operations $\Fore[i][j]$ (respectively, $\Back[i][j]$)  allows to recover the location in column $j+1$  (respectively, in column $j-1$) of the row index  in entry $[i][j]$ of the prefix-array. 

\smallskip
\noindent
\textbf{Forward and backward stepping on PBWT.}
We define $\Fore[i][j]$, for $i\in [1..\height]$ and $j\in[1..\width)$, as the (row) index of $\col_j(\PA)[i]$ in $\col_{j+1}(\PA)$, and $\Back[i][j]$, for $i\in [1..\height]$ and $j\in(1..\width]$, as the (row) index of $\col_j(\PA)[i]$ in $\col_{j-1}(\PA)$.

Previously, the operation $\Fore[i][j]$ could be implemented as follows. Let $C_c$ denote the number of occurrences of symbols $c' < c$ in column $j$ of the $\PBWT$, that is, in $\col_j(\PBWT)$. Then, $\Fore[i][j] = C_c + \rank_c(\col_j(\PBWT), i)$, where $c = \col_j(\PBWT)[i]$. The backward stepping operation can be implemented in a symmetric manner.

% Previously, the operation $\Fore[i][j]$ can be implemented as follows:
% Let $C_c$ be the number of occurrences of symbols $c'<c$ in the column $j$ of $\PBWT$, that is, $\col_j(\PBWT)$.
% Then, $\Fore[i][j]=C_c + \rank_c(\col_j(\PBWT), i)$, where $c=\col_j(\PBWT)[i]$.
% The backward stepping operation can be implemented symmetrically.

%Proposition \ref{prop-lf-map} states the properties relevant to $\Fore$ , whose proof is deferred to Appendix \ref{app-prop-lf-map}.
Proposition \ref{prop-lf-map} states the key properties of $\Fore$, with the proof deferred to Appendix \ref{app-prop-lf-map}.
 
\begin{proposition}\label{prop-lf-map}
    (a) If $\col_j(\PBWT)[i']=\col_j(\PBWT)[i'']$ and $i'<i''$, then $\Fore[i'][j]<\Fore[i''][j]$, and (b) if $\col_j(\PBWT)[i']=\col_j(\PBWT)[i'+1]$, then $\Fore[i'][j]+1=\Fore[i'+1][j]$.
\end{proposition}

% \begin{proof}
% When constructing the $j$-th column of $\PBWT$, note that the prefixes of the haplotypes up to column $j-1$ are stably sorted.
% Therefore, statements (a) and (b) follow directly from this stable ordering.
% \end{proof}

\section{A Lower Bound and an Upper Bound on $\newR$}
\label{sect-bounds-newR}

In this section, we establish lower and upper bounds on $\newR$ in terms of $\height''$, where $\height''$ denotes the number of pairs $(S_i, S_{i+1})$ such that $1 \le i < \height$ and $S_i \ne S_{i+1}$.
%All proofs omitted in the section can be found in Appendix \ref{app-bounds-newR}.

\begin{lemma}\label{lem-height-r-2}
It holds that $\newR \ge \height'' + 1$.
\end{lemma}

\begin{proof}
Consider any pair $(S_i, S_{i+1})$ such that $S_i \ne S_{i+1}$. 
Let $j''$ denote the length of the longest common prefix between $S_i$ and $S_{i+1}$.
Observe that the indices $i$ and $i+1$ remain consecutive in $\col_j(\PA)$ for all $1 \le j \le j'' + 1$.
Let $\tau$ be the row index such that $\col_{j''+1}(\PA)[\tau] = i$.
By the observation above, we have $\col_{j''+1}(\PA)[\tau + 1] = i + 1$.
Since $\col_{j''+1}(\PBWT)[\tau] \ne \col_{j''+1}(\PBWT)[\tau + 1]$ by the definition of $j''$, there must be a run boundary between rows $\tau$ and $\tau + 1$ in the $(j'' + 1)$-st column of $\PBWT$.
Therefore, each pair $(S_i, S_{i+1})$ with $S_i \ne S_{i+1}$ corresponds to a run boundary in $\PBWT$.
The mapping from such a pair to a run boundary at $(\tau, j''+1)$ is injective, since each haplotype index $i$ appears at a unique row $\tau$ in $\PA$ of any column.
This establishes that $\newR \ge \height'' + 1$.
\end{proof}

\begin{corollary}\label{prop-height-r}
It follows that $\newR$ is at least the number of distinct haplotypes in $\{S_1, S_2, \dots, S_{\height}\}$.
\end{corollary}

\begin{proof}
    The statement holds because the number of distinct haplotypes in $\{S_1, S_2, \dots, S_{\height}\}$ is at most $\height'' + 1$.
\end{proof}

We call an interval $[b, e]$ \emph{canonical} with respect to column $j$ of the PBWT, for $1 \le j \le \width$,  if $[b, e]$ is the maximal interval such that $S_{\col_j(\PA)[b]} = S_{\col_j(\PA)[b+1]} = \cdots = S_{\col_j(\PA)[e]}$.
Let $\ell_j$ denote the number of canonical intervals with respect to column $j$ for $1 \le j \le \width$.
Lemma~\ref{lemma-r-ell} describes the relationship between $\ell_j$ and $r_j$.

\begin{lemma}\label{lemma-r-ell}
It holds that $r_j \le \ell_j$.
\end{lemma}

\begin{proof}
    Consider column $j$ for any $1 \le j \le \width$.
    Let $[b, e]$ denote any canonical interval with respect to this column.
    Observe that $\col_j(\PBWT)[b] = \col_j(\PBWT)[b+1] = \cdots = \col_j(\PBWT)[e]$.
    Hence, $[b, e]$ corresponds to a contiguous block of identical symbols in $\col_j(\PBWT)$, although the block might not be maximal.
    This implies that $r_j \le \ell_j$.
\end{proof}

%We next establish an upper bound on $\ell_j$.

\begin{lemma}\label{lem-bound-ell_j}
It holds that $\ell_j \le \height'' + 1$.
\end{lemma}

\begin{proof}
    The proof proceeds by induction on $j$.
    Let $H$ be the array consisting of $\{\height\} \cup \{\, i \mid S_i \ne S_{i+1} \text{ and } 1 \le i < \height \,\}$ in increasing order.
    It follows that $|H| = 1 + \height''$.

    For the base case $j = 1$, observe that each interval $[H[t-1] + 1, H[t]]$ for $1 \le t \le |H|$ (under the convention $H[0] = 0$) forms a distinct canonical interval with respect to the first column.
    Since $\bigcup_t [H[t-1] + 1, H[t]] = [1, \height]$, we have $\ell_1 = |H|$.

    Assume inductively that $\ell_{j-1} \le \height'' + 1$.
    Let $[b_{j-1}, e_{j-1}]$ denote any canonical interval with respect to column $j-1$, so we have $\col_{j-1}(\PBWT)[b_{j-1}] = \cdots = \col_{j-1}(\PBWT)[e_{j-1}]$.
    By Proposition~\ref{prop-lf-map}(b), the integers $\Fore[b_{j-1}][j-1], \Fore[b_{j-1}+1][j-1], \dots, \Fore[e_{j-1}][j-1]$ are consecutive, forming an interval $[\Fore[b_{j-1}][j-1], \Fore[e_{j-1}][j-1]]$.
    Furthermore, by the definitions of $\Fore$ and $[b_{j-1}, e_{j-1}]$, the haplotypes $S_{\col_j(\PA)[t]}$ for all $t \in [\Fore[b_{j-1}][j-1], \Fore[e_{j-1}][j-1]]$ are identical.
    Thus, this interval is contained within some canonical interval with respect to column $j$.
    This containment may be strict whenever multiple canonical intervals with respect to column $j-1$, corresponding to identical haplotype sequences, are mapped to adjacent positions in column $j$, thereby merging into a single canonical interval with respect to column $j$.
    Therefore, each canonical interval with respect to column $j-1$ is mapped by $\Fore$ to a subinterval of a canonical interval with respect to column $j$.
    Since the union of these mapped intervals covers all rows $[1, \height]$, the number of canonical intervals in column $j$ must be less than or equal to the number in column $j-1$ ($\ell_j \le \ell_{j-1}$).
    By the induction hypothesis, $\ell_j \le \ell_{j-1} \le \height'' + 1$, completing the proof.
\end{proof}

%Finally, we conclude this section by establishing an upper bound on $\newR$.

\begin{theorem}\label{theorem-upper-newR}
It holds that $\newR \le \width \cdot (\height'' + 1)$.
\end{theorem}

\begin{proof}
By Lemmas~\ref{lemma-r-ell} and~\ref{lem-bound-ell_j}, we have $r_j \le \height'' + 1$ for all $j$.
Therefore, $\newR = \sum_{1 \le j \le \width} r_j 
    \le \sum_{1 \le j \le \width} (\height'' + 1)
    = \width (\height'' + 1).$
\end{proof}

%\section{The Three-Overlap Constraint and a Three-Overlap Normalization Algorithm}
\section{The Three-Overlap Constraint and a Normalization Algorithm}
\label{sect-three-overlap-normal}

In this section, we define the three-overlap constraint and present a new algorithm for partitioning runs based on it.

% \begin{definition}[The Three-Overlap Constraint]
% Let $I_p=[p_1, p'_1], [p_2, p'_2], \dots [p_x, p'_x]$ be $x$ pairwise-disjoint intervals into which $[1, n]$ is partitioned, where $p_1=1, p'_x=n$, and $p'_i+1=p_{i+1}$ for $1\le i < x$.
% Let $I_q=[q_1, q'_1], [q_2, q'_2], \dots [q_y, q'_y]$ be $y$ pairwise-disjoint intervals into which $[1, n]$ is partitioned, where $q_1=1, q'_y=n$, and $q'_j+1=q_{j+1}$ for $1\le j < y$.
% $I_p$ satisfies the three-overlap constraint with respect to $I_q$, each interval in $I_p$ overlaps at most three intervals in $I_q$.
% \end{definition}

\begin{definition}[The Three-Overlap Constraint]
Let $I_p = \{ [p_1, p'_1], [p_2, p'_2], \dots, [p_x, p'_x] \}$
be a collection of $x$ pairwise disjoint intervals that partition the range $[1, n]$, 
where $p_1 = 1$, $p'_x = n$, and $p_{i+1}=p'_i + 1 $ for $1 \le i < x$.
Similarly, let $I_q = \{ [q_1, q'_1], [q_2, q'_2], \dots, [q_y, q'_y] \}$ be a collection of $y$ pairwise disjoint intervals that also partition $[1, n]$, 
where $q_1 = 1$, $q'_y = n$, and $q_{j+1}=q'_j + 1 $ for $1 \le j < y$.
We say that $I_p$ \emph{satisfies the three-overlap constraint with respect to} $I_q$
if every interval in $I_p$ overlaps with at most three intervals in $I_q$.
\end{definition}

Note that the three-overlap constraint differs from the \emph{balancing} property introduced in the move data structure \cite{NishimotoT21} in two key aspects. 
Under the balancing property, 
\begin{itemize}
    \item $I_p$ and $I_q$ are in bijective correspondence---that is, there exists a bijection $f(\cdot)$ (with inverse $f^{-1}(\cdot)$) such that for each interval $\mathsf{I} \in I_p$ (resp., $\mathsf{I} \in I_q$), the interval $[f(\mathsf{I}.b), f(\mathsf{I}.e)]$ belongs to $I_q$ (resp., $[f^{-1}(\mathsf{I}.b), f^{-1}(\mathsf{I}.e)]\in I_p$), and
    \item each interval in $I_p$ contains at most three left endpoints of intervals in $I_q$; consequently, a single interval in $I_p$ may overlap with up to four distinct intervals from $I_q$.
\end{itemize}

% . 
% Second, each interval in $I_p$ contains at most three left endpoints of intervals in $I_q$; consequently, a single interval in $I_p$ may overlap with up to four distinct intervals from $I_q$.

% \begin{figure}%[!h]
% 		\centering
% 		%\captionsetup{justification=centering}
% 		\includegraphics[scale=1]{./figs/normal-alg}
% 		\caption{
% Example of the normalization algorithm. 
% Given $I_p = \{[1,1], [2,11], [12,16]\}$ and 
% $I_q = \{[1,2], [3,3], [4,5], [6,7], [8,9], [10,10], [11,13], [14,14], [15,16]\}$, 
% the algorithm outputs $\hat{I}_p = \{[1,1], [2,5], [6,10], [11,11],  [12,16]\}$, 
% where each interval in $\hat{I}_p$ overlaps at most three intervals in $I_q$.
% }
%     \label{fig:normal-alg}
% \end{figure}

\begin{lemma}\label{lem-interval-overlap-three}
Let $I_p$ and $I_q$ denote two lists of intervals into which $[1,n]$ is partitioned.
There exists an $O(|I_p|+|I_q|)$-time algorithm that partitions all intervals in $I_p$ into at most $(|I_p|+\Big\lfloor \frac{|I_q|}{2}\Big\rfloor)$ sub-intervals, satisfying the three-overlap constraint, with respect to $I_q$.
%Let $I_p=[p_1, p'_1], [p_2, p'_2], \dots [p_x, p'_x]$ be $x$ pairwise-disjoint intervals into which $[1, n]$ is partitioned, where $p_1=1, p'_x=n$, and $p'_i+1=p_{i+1}$ for $1\le i < x$.
%Let $I_q=[q_1, q'_1], [q_2, q'_2], \dots [q_y, q'_y]$ be $y$ pairwise-disjoint intervals into which $[1, n]$ is partitioned, where $q_1=1, q'_y=n$, and $q'_j+1=q_{j+1}$ for $1\le j < y$, where $y\ge 4$.
%One can partition all intervals in $I_p$ into at most $(x+\lfloor \frac{y}{2}\rfloor)$ sub-intervals such that every sub-interval overlaps with at most three intervals in $I_q$.
\end{lemma}
\begin{proof}
Let $x=|I_p|$ and $I_p=[p_1, p'_1], [p_2, p'_2], \dots, [p_x, p'_x]$ be $x$ pairwise-disjoint intervals into which $[1, n]$ is partitioned, where $p_1=1, p'_x=n$, and $p'_i+1=p_{i+1}$ for $1\le i < x$.
Let $y=|I_q|$.
Without loss of generality, assume that $y\ge 4$; otherwise, $I_p$ immediately satisfies the three-overlap constraint.
Let $I_q=[q_1, q'_1], [q_2, q'_2], \dots, [q_y, q'_y]$ be $y$ pairwise-disjoint intervals into which $[1, n]$ is partitioned, where $q_1=1, q'_y=n$, and $q'_j+1=q_{j+1}$ for $1\le j < y$.

%We first describe the algorithm.
The algorithm $\texttt{normalization}(I_p, I_q)$ is described as follows.
Create a variable $k$ and initiate $k$ to $1$.
Iterate and process each interval $[p_i, p'_i]$ in $I_p$ from left to right as follows:
If $[p_i, p'_i]$ overlaps with at most three intervals of $I_q$, then skip $[p_i, p'_i]$ and move on to the next interval;
otherwise, divide $[p_i, p'_i]$ into two sub-intervals, $[p_i, d_k]$ and $[d_k+1, p'_i]$, where $1\le d_k \le n$ is the largest integer such that $[p_i, d_k]$ overlaps three intervals of $I_q$.
Then, substitute $[p_i, p'_i]$ for $[p_i, d_k]$ and $[d_k+1, p'_i]$, increment $k$ by one, and move on to the next interval, that is, $[d_k+1, p'_i]$.

Let $\hat{I}_p$ denote the list of intervals outputted by the above algorithm.
Figure~\ref{fig:normal-alg} (in Appendix \ref{app:normal-alg}) illustrates an example of the algorithm.
Clearly, we have $|\hat{I}_p|=x+k$ and it follows that every interval in $\hat{I}_p$ overlaps at most three intervals in $I_q$.
It remains to prove that $k\le \big\lfloor \frac{y}{2}\big\rfloor$.
To this end, we define $\Lambda(\hat{I}_p[i])$ for any $1\le i\le x+k$ to be the set of intervals in $I_q$, overlapping $\hat{I}_p[i]$.
%The algorithm above guarantees $|\Lambda(\hat{I}_p[i])|\le 3$.

\begin{claim}\label{claim-lambda-intersect}
    For any $1\le t_1<t_2<k$, let $\mathsf{I}_1$ (resp. $\mathsf{I}_2$) denote the interval in $\hat{I}_p$, of which the right endpoint is $d_{t_1}$ (resp. $d_{t_2}$).
    We have $\Lambda(\mathsf{I}_1)\cap \Lambda(\mathsf{I}_2)=\emptyset$.
\end{claim}

To see the correctness of the claim, observe that we have $d_{\tau}\in \{q'_1, \dots, q'_y\}$ for any $1\le \tau<k$, and thus the integers $d_{\tau}$ and $d_{\tau}+1$ are always in different intervals in $I_q$.
Moreover, $d_{t_1}+1\notin \tilde{\mathsf{I}}$ for any $\tilde{\mathsf{I}}\in \Lambda(\mathsf{I}_1)$, in view of the algorithm.

Assume that there exists an interval $\mathsf{I}\in \Lambda(\mathsf{I}_1)\cap \Lambda(\mathsf{I}_2)$.
Let $z$ denote the left endpoint of the interval $\mathsf{I}_2$.
Then, we have $d_{t_1}\in \mathsf{I}$ and $z\in \mathsf{I}$.
Since $d_{t_1}<d_{t_1}+1\le z$, we have $d_{t_1}+1\in \mathsf{I}\in \Lambda(\mathsf{I}_1)$ as well, a contradiction.
Hence, the assumption is false. $\diamond$

In view of the algorithm, the interval in $\hat{I}_p$ ending at $d_{\tau}$ for any $1\le \tau<k$ overlaps exactly three intervals in $I_q$.
In view of Claim \ref{claim-lambda-intersect}, we have $k-1\le \big\lfloor \frac{y}{3} \big\rfloor$; therefore, $x+k\le x+1+\big\lfloor \frac{y}{3} \big\rfloor\le x+\big\lfloor \frac{y}{2}\big\rfloor$, for $y\ge 4$.   

Clearly, the algorithm runs in $O(|\hat{I}_p|+|I_p|+|I_q|)\subseteq O(|I_p|+|I_q|)$ time.
\end{proof}

\section{Constructing Sub-Runs and Bounding Their Number}
\label{sect-sub-run-back}

In this section, we introduce \emph{sub-runs} for forward and backward stepping (i.e., $\Fore$ and $\Back$) in the PBWT.

While runs are maximal-length substrings consisting of the same character, we define sub-runs as substrings of the same character without the maximal-length restriction.

We define a \emph{run interval} with respect to column $j$ ($1 \le j \le \width$) as the maximal interval $[b, e]$ such that all symbols in positions $b, b+1, \dots, e$ of $\col_j(\PBWT)$ are identical. 
Let $\inter_j$ denote the list of the $r_j$ run intervals in column $j$, sorted in increasing order of their starting positions.
Similarly, a \emph{sub-run interval} is any contiguous subrange contained within a run interval.
We abuse notation slightly and use the terms run intervals (resp. sub-run intervals) and runs (resp. sub-runs) interchangeably.

% In the following, we define two bijection functions $\BiFore$ and $\BiBack$ that will be used, respectively, in the construction of two lists $\subIB_j$ and  $\subIF_j$ of sub-intervals for $1\le j\le \width$, which are essentially sub-runs for column $j$:
In the following, we define two bijection functions $\BiFore$ and $\BiBack$ that will be used, respectively, in the construction of two lists $\subIB_j$ and  $\subIF_j$ of sub-runs for each column $j$.
The lists $\subIB_j$ and $\subIF_j$ are used later to implement $\Back$ queries and $\Fore$ queries, respectively. 
%The two bijctions are functions $\BiFore$ and $\BiBack$. 
More precisely,  $\BiFore_j(L)$ returns the list $\{ [\Fore[b_\tau][j], \Fore[e_\tau][j]] \mid 1 \le \tau \le |L| \}$, sorted in increasing order by the left endpoint of each interval, for any set of intervals $L = \{ [b_\tau, e_\tau] \mid 1 \le \tau \le |L|, 1 \le b_\tau \le e_\tau \le \height \}$ and any $1 \le j <\width$.

 By Proposition~\ref{prop-lf-map}(b), it follows that $\Fore[i+1][j] - \Fore[i][j] = 1$ for any $1 \le \tau \le r_j$ and $p_\tau \le i < p_{\tau+1}-1$. Hence, $\BiFore_j(L)$ is well-defined (i.e., $\Fore[e_\tau][j] \ge \Fore[b_\tau][j]$ for every $1 \le \tau \le |L|$) if each interval in $L$ is a sub-interval of some interval in $\inter_j$.  

Symmetrically, we define $\BiBack_j(L)$ to return the list $\{ [\Back[b_\tau][j], \Back[e_\tau][j]] \mid 1 \le \tau \le |L| \}$, sorted in increasing order by the left endpoint of each interval.  
%Recall that  $\Back$ and $\Fore$ are the main forward and backward stepping in the PBWT.

\noindent
\textbf{Constructing sub-runs for $\Back$ queries.}  
The construction of $\subIB_j$ proceeds by induction. In the base case, we set $\subIB_1 = \inter_1$.  
For each $j$ from $2$ to $\width$, we construct $\subIB_j$ as follows: we apply the algorithm \texttt{normalization}$(I_p, I_q)$ (see Lemma~\ref{lem-interval-overlap-three}) with $I_p := \inter_j$ and $I_q := \BiFore_{j-1}(\subIB_{j-1})$. 
Recall that the function $\BiFore_{j-1}$  maps sub-intervals in column $j-1$ to sub-intervals in column $j$. 
The list of sub-intervals output by this algorithm is assigned to $\subIB_j$.  
See Figure~\ref{fig:back} for an illustration of the algorithm and an example.

\begin{figure}[!t]
  \centering
  \begin{subfigure}[b]{0.44\textwidth}
    \centering
    \includegraphics[width=\textwidth]{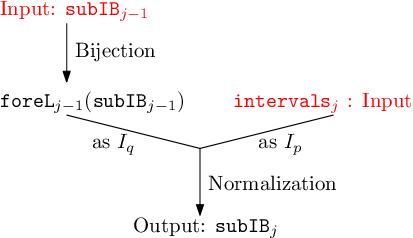}
    \caption{The scheme of building $\subIB_j$}
  \end{subfigure}
  \hfill
  \begin{subfigure}[b]{0.54\textwidth}
    \centering    \includegraphics[width=\textwidth]{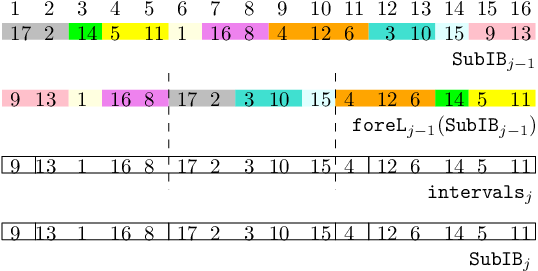}
    \caption{An Example of building $\subIB_j$}
  \end{subfigure}
  \caption{\label{fig:back} Illustration of the algorithm scheme for building sub-runs in $\subIB_j$ and an example. In this example, $\subIB_{j-1} = \{\,[1,2], [3,3], [4,5], [6,6], [7,8], [9,11], [12,13], [14,14], [15,16]\,\}$ and 
$\inter_j = \{\,[1,1], [2,11], [12,16]\,\}$.  
The bijective function $\BiFore_{j-1}$ maps the list $\subIB_{j-1}$ to the list 
$\{\,[1,2], [3,3], [4,5], [6,7], [8,9], [10,10], [11,13], [14,14], [15,16]\,\}$.  
Intervals highlighted in the same color contain the same haplotype indices and indicate corresponding pairs under this bijection.  
After applying the normalization algorithm to $\inter_j$ and $\BiFore_{j-1}(\subIB_{j-1})$, the intervals in $\inter_j$ are partitioned into  
$\subIB_j = \{\,[1,1], [2,5], [6,10], [11,11], [12,16]\,\}$.  
Each interval in $\subIB_j$ overlaps with at most three intervals in $\BiFore_{j-1}(\subIB_{j-1})$.}
\end{figure}

In view of the algorithm in Lemma \ref{lem-interval-overlap-three}, each interval in $\subIB_j$ for $1\le j\le \width$ is a sub-interval of some interval in $\inter_j$ and overlaps at most three intervals in $\BiFore_{j-1}(\subIB_{j-1})$.
Recall that every interval in $\inter_j$ corresponds to a run in $\col_j(\PBWT)$.
Hence, every interval in $\subIB_j$ corresponds to a sub-run in $\col_j(\PBWT)$.

\smallskip
\noindent
\textbf{Constructing sub-runs for $\Fore$ queries.}  
The construction of $\subIF_j$ is also performed by induction. In the base case, we set $\subIF_{\width} = \inter_{\width}$.  
Then, for $j$ from $\width - 1$ down to $1$, we construct $\subIF_j$ as follows:  
we apply the algorithm \texttt{normalization}$(I_p, I_q)$ (see Lemma~\ref{lem-interval-overlap-three}), where  
$I_p = \BiFore_j(\inter_j)$ and $I_q = \subIF_{j+1}$.  
Let $\subIF'_{j+1}$ be the list of intervals output by the algorithm; we then assign $\subIF_j = \BiBack_{j+1}(\subIF'_{j+1})$.  
Figure~\ref{fig:fore} (in Appendix \ref{app-fore}) illustrates the construction scheme along with an example.

% \begin{figure}[!t]
%   \centering
%   \begin{subfigure}[b]{0.44\textwidth}
%     \centering
%     \includegraphics[width=\textwidth]{./figs/fore}
%     \caption{The scheme of building $\subIF_j$}
%   \end{subfigure}
%   \hfill
%   \begin{subfigure}[b]{0.54\textwidth}
%     \centering
%     \includegraphics[width=\textwidth]{./figs/building-fore-runs-2}
%     \caption{An Example of building $\subIF_j$}
%   \end{subfigure}
%   \caption{\label{fig:fore} Illustration of the algorithm scheme for building sub-runs in $\subIF_j$ and an example. In this example, $\subIF_{j+1} = \{\,[1,2], [3,3], [4,5], [6,7], [8,9], [10,10], [11,13], [14,14], [15,16]\,\}$ and 
% $\inter_j = \{\,[1,5], [6,6], [7,16]\,\}$.  
% The bijective function $\BiFore_{j}$ maps the list $\inter_{j}$ to the list 
% $\{\,[1,1], [2,11], [12, 16]\,\}$.  
% Intervals highlighted in the same color contain the same haplotype indices and indicate corresponding pairs under this bijection.  
% After applying the normalization algorithm to $\BiFore_j(\inter_j)$ and $\subIF_{j+1}$, the intervals in $\BiFore_j(\inter_j)$ are partitioned into  
% $\subIF'_j = \{\,[1,1], [2,5], [6,10], [11,11], [12,16]\,\}$.  
% Each interval in $\subIF'_{j+1}$ overlaps with at most three intervals in $\subIF_{j+1}$.
% Finally, by setting \(\subIF_j = \BiBack_{j+1}(\subIF'_{j+1})\), the intervals \(\inter_j\) are partitioned into \(\subIF_j = \{[1,5], [6,6], [7,10], [11,15], [16,16]\}\).
% }
% \end{figure}

By the normalization algorithm shown in Lemma \ref{lem-interval-overlap-three}, each interval in $\subIF'_{j+1}$ for $1\le j\le \width$ is a sub-interval of some interval in $\BiFore_j(\inter_j)$ and overlaps at most three intervals in $\subIF_{j+1}$.
Recall that every interval in $\BiFore_j(\inter_{j})$ corresponds to a run in $\col_j(\PBWT)$.
Likewise, every interval in $\subIF'_{j+1}$ corresponds to a sub-run in $\col_j(\PBWT)$.
Since $\subIF_j = \BiBack_{j+1}(\subIF'_{j+1})$, every interval in $\subIF_j$ is a sub-run in $\col_j(\PBWT)$.

% \medskip
% \noindent
% \textbf{Bound the number of sub-runs created.}
% Let $\newR=\sum_{1\le j\le \width} r_j$ denote the total number of runs across all $\width$ columns of the matrix $\PBWT$.
% Next, we bound the total number of sub-runs to create.

\begin{lemma}\label{lem-size-sub-runs}
    We have $\sum_{1\le j\le \width} |\subIB_j|< 2\newR$ and $\sum_{1\le j\le \width} |\subIF_j|< 2\newR$.
\end{lemma}
\begin{proof}
    Observe that it follows that $|\BiFore_j(\subIB_{j})|=|\subIB_{j}|$.
    Recall that $|\inter_j|=r_j$ is always true.
    In the base case, we have $|\subIB_1|=r_1$.
    For $1<j\le \width$, we have the recursion $|\subIB_{j}|\le |\inter_j|+\lfloor \frac{|\subIB_{j-1}|}{2}\rfloor\le r_j + \frac{|\subIB_{j-1}|}{2}$, in view of Lemma \ref{lem-interval-overlap-three} and the observation above.
    By solving the recursion, it follows that $|\subIB_j|=\frac{r_1}{2^{j-1}}+\frac{r_2}{2^{j-2}}+\dots+ \frac{r_j}{2^{j-j}}=\sum_{1\le\tau\le j} \frac{r_{\tau}}{2^{j-\tau}}$.
    Therefore, we have $\sum_{1\le j\le \width} |\subIB_j|\le \sum_{1\le j\le \width} \sum_{1\le\tau\le j} \frac{r_{\tau}}{2^{j-\tau}}< 2(r_1+r_2+\dots+r_\width)= 2\newR$.
    The bound on the total number of $|\subIF_j|$ for all $j$ can be computed similarly.
    This concludes the proof.
\end{proof}

% \begin{figure}[!t]
%     \centering
%     %\captionsetup{justification=centering}
%     \includegraphics[scale=1]{./figs/building-fore-runs}
%     \caption{
% Example illustrating the construction of $\subIF_j$ from $\subIF_{j+1}$ and $\inter_j$.  
% In this example, $\subIF_{j+1} = \{\,[1,2], [3,3], [4,5], [6,7], [8,9], [10,10], [11,13], [14,14], [15,16]\,\}$ and 
% $\inter_j = \{\,[1,5], [6,6], [7,16]\,\}$.  
% The bijective function $\BiFore_{j}$ maps the intervals in $\inter_{j}$ to the list 
% $\{\,[1,1], [2,11], [12,16]\,\}$.  
% Intervals highlighted in the same color indicate corresponding pairs under this bijection.  
% After applying the normalization algorithm to $\BiFore_{j}(\inter_j)$ and $\subIF_{j+1}$, the intervals in $\BiFore_{j}(\inter_j)$ are partitioned into  
% $\subIF'_j = \{\,[1,1], [2,5], [6,10], [11,11], [12,16]\,\}$.  
% Each interval in $\subIF'_j$ overlaps with at most three intervals in $\subIF_{j+1}$.  
% By applying the inverse bijection, the intervals in $\inter_j$ are partitioned into  
% $\subIF_j = \{\,[1,5], [6,6], [7,10], [11,15], [16,16]\,\}$.
% }
%     \label{fig:building-runs-fore}
%   \end{figure}

\section{The Data Structure for Constant-Time $\Fore/\Back$ Queries}
\label{sect-ds-fore-back}

In this section, we design the data structure constructed over these sub-runs to support $\Fore/\Back$ queries and access to entries in the matrix $\PBWT$.

\begin{theorem}\label{theorem-PBWT-r}
There exists a data structure of $O(\newR)$ words, constructed over $\subIB_j$ and $\subIF_j$ for $1 \le j \le \width$, that supports each of the following operations in constant time---without accessing the original matrix $M$, its $\PBWT$, or its prefix arrays $\PA$:
\begin{itemize}
    \item Given an index $i' \in [1..\height]$ and the index of the interval in $\subIB_{j'}$ that contains $i'$, one can find $\Back[i'][j']$, determine the index of the interval in $\subIB_{j'-1}$ containing $\Back[i'][j']$, and retrieve $\col_{j'}(\PBWT)[i']$.
    \item Given an index $i' \in [1..\height]$ and the index of the interval in $\subIF_{j'}$ that contains $i'$, one can find $\Fore[i'][j']$, determine the index of the interval in $\subIF_{j'+1}$ containing $\Fore[i'][j']$, and retrieve $\col_{j'}(\PBWT)[i']$.
\end{itemize}
\end{theorem}

\begin{proof}
We first describe the data structure for $\Back$ queries.
For each column $j$ of the \PBWT, we construct the list $\subIB_j$ as described in Section~\ref{sect-sub-run-back} (see Constructing sub-runs for back queries).
Observe that each interval in $\subIB_j$ overlaps at most three intervals in $\BiFore_{j-1}(\subIB_{j-1})$.  Moreover, each interval $\tilde{\mathsf{I}} \in \BiFore_{j-1}(\subIB_{j-1})$ corresponds to the interval $[\Back[\tilde{\mathsf{I}}.b][j], \Back[\tilde{\mathsf{I}}.e][j]]$ belonging to $\subIB_{j-1}$.
For every $1 < j \le \width$ and each interval $\subIB_j[\tau]$ with $1 \le \tau \le |\subIB_j|$, we store a list $B^{\tau}_{j}$ of quadruples $(\tilde{s}, \tilde{t}, s, \lambda)$, where:
\begin{itemize}
    \item $[\tilde{s}, \tilde{t}]$ is a distinct interval in $\BiFore_{j-1}(\subIB_{j-1})$ that overlaps $\subIB_j[\tau]$;
    \item $s = \Back[\tilde{s}][j]$; and
    \item $\lambda$ is the index of the interval in $\subIB_{j-1}$ whose left endpoint is $s$.
\end{itemize}
Since each $\subIB_j[\tau]$ overlaps at most three intervals in $\BiFore_{j-1}(\subIB_{j-1})$, the list $B^{\tau}_{j}$ contains at most three tuples.
%After building all lists $B^{\tau}_{j}$, we discard the lists $\subIB_j$ to save space.
By Lemma~\ref{lem-size-sub-runs}, $\sum_j |\subIB_j| < 2\newR$, so the total space usage is $O(\newR)$ words.

Let $x$, given in a query, denote the index of the sub-run in $\subIB_{j'}$ containing $i'$.
To answer a query $\Back[i'][j']$, we search in the list $B^x_{j'}$ for the quadruple $(\tilde{s}, \tilde{t}, s, \lambda)$ satisfying $i' \in [\tilde{s}, \tilde{t}]$.
Then $\Back[i'][j'] = i' - \tilde{s} + s$, and $\lambda$ is the index of the interval in $\subIB_{j'-1}$ containing $\Back[i'][j']$.
We return $(i' - \tilde{s} + s, \lambda)$ as the result.
Since $|B^x_{j'}| \le 3$, each query takes $O(1)$ time.
See Figure~\ref{fig:ds-runs-back} for an illustration of the data structure and algorithm.

\begin{figure}[!t]
    \centering
    %\captionsetup{justification=centering}
    \includegraphics[scale=1]{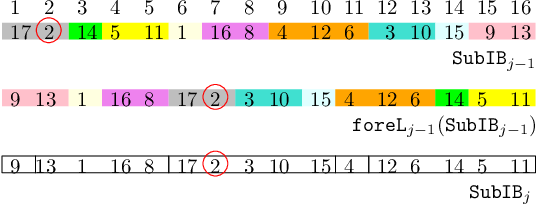}
\caption{
Example illustrating the data structure and algorithm for $\Back$ queries.  
The third interval $[6,10]$ in $\subIB_j$ overlaps three intervals in $\BiFore_{j-1}(\subIB_{j-1})$: $[6,7]$, $[8,9]$, and $[10,10]$.  
These correspond to the first $[1,2]$, seventh $[12,13]$, and eighth $[14,14]$ intervals in $\subIB_{j-1}$, respectively.  
Thus, the data structure$B^3_j$ stores the quadruples $\{(6,7,1,1), (8,9,12,7), (10,10,14,8)\}$ in the form $(\tilde{s}, \tilde{t}, s, \lambda)$.  
For a query $\Back[7][j]$ with index $3$, the algorithm finds $(\tilde{s}:6, \tilde{t}:7, s:1, \lambda:1)$ in $B^3_j$ (since $7 \in [6,7]$) and returns $7-\tilde{s}+s = 2$ and $\lambda = 1$.
}
    \label{fig:ds-runs-back}
  \end{figure}

\medskip
\noindent\textbf{Accessing $\col_{j'}(\PBWT)[i']$ from $\subIB$ lists.}
We store, for each $1 \le j \le \width$, an array $\valB_j$ of length $|\subIB_j|$, where each entry $\valB_j[\tau]$ is set to $\col_j(\PBWT)[\subIB_j[\tau].b]$.
This requires $O(r_j)$ words per column, and $O(\newR)$ words overall.
Given the index $x$ of the sub-run containing $i'$ in $\subIB_{j'}$, we have $\col_{j'}(\PBWT)[i'] = \valB_{j'}[x]$, which can be retrieved in $O(1)$ time.

\medskip
\noindent\textbf{Data structure for $\Fore$ queries.}
We construct the lists $\subIF_j$ for each column $j$ as described in Section~\ref{sect-sub-run-back} (see Constructing sub-runs for fore queries).
Each interval $\mathsf{I} \in \subIF_j$ corresponds to $[\Fore[\mathsf{I}.b][j], \Fore[\mathsf{I}.e][j]] \in \BiFore_j(\subIF_j)$, and each interval in $\BiFore_j(\subIF_j)$ overlaps at most three intervals in $\subIF_{j+1}$.
For every $1 \le j < \width$ and each interval $\subIF_j[\tau]$, we store a list $F^{\tau}_{j}$ of quintuplets $(s', \tilde{s}, s, t, \lambda)$, where:
\begin{itemize}
    \item $s' = \subIF_j[\tau].b$, $\tilde{s} = \Fore[s'][j]$,
    %\item $\tilde{s} = \Fore[s'][j]$,
    \item $[s, t]$ is a distinct interval in $\subIF_{j+1}$ overlapping $[\Fore[\subIF_j[\tau].b][j], \Fore[\subIF_j[\tau].e][j]]$,
    \item $\lambda$ is the index of $[s, t]$ in $\subIF_{j+1}$.
\end{itemize}
Since each such interval, i.e., $[\Fore[\subIF_j[\tau].b][j], \Fore[\subIF_j[\tau].e][j]]$, overlaps at most three intervals in $\subIF_{j+1}$, each $F^{\tau}_{j}$ has at most three tuples.
%After building all lists $F^{\tau}_{j}$, we discard the lists $\subIF_j$.
By Lemma~\ref{lem-size-sub-runs}, the total space is again $O(\newR)$ words.

Let $x$ denote the index of the sub-run containing $i'$ in $\subIF_{j'}$.
To answer $\Fore[i'][j']$, we search $F^x_{j'}$ for the quintuple $(s', \tilde{s}, s, t, \lambda)$ satisfying $i' - s' + \tilde{s} \in [s, t]$.
Then $\Fore[i'][j'] = i' - s' + \tilde{s}$, and $\lambda$ is the index of the interval in $\subIF_{j'+1}$ containing $\Fore[i'][j']$.
We return $(i' - s' + \tilde{s}, \lambda)$ as the result.
Since $|F^x_{j'}| \le 3$, each query takes $O(1)$ time.
An illustration of the data structure and the algorithm is depicted in Figure~\ref{fig:ds-runs-fore} (in Appendix~\ref{app-ds-runs-fore}).

% \begin{figure}[!t]
%     \centering
%     \includegraphics[scale=1]{./figs/ds-fore-runs-2}
% \caption{Example illustrating the data structure and algorithm for $\Fore$ queries. 
% The fourth interval $[11,15]$ in $\subIF_j$ corresponds to the third interval $[6,10]$ in $\BiFore_j(\subIF_j)$, which overlaps three intervals in $\subIF_{j+1}$: the fourth $[6,7]$, the fifth $[8,9]$, and the sixth $[10,10]$. 
% Accordingly, the data structure $F^4_j$ built for the interval $[11,15]$ stores the quintuples $\{(11,6,6,7,4), (11,6,8,9,5), (11,6,10,10,6)\}$ in the form $(s', \tilde{s}, s, t, \lambda)$. 
% Given a query $\Fore[14][j]$ with index $4$, the algorithm finds the tuple $(s'=11, \tilde{s}=6, s=8, t=9, \lambda=5)$ in $F^4_j$ (since $14 - s' + \tilde{s} = 14 - 11 + 6 = 9 \in [8,9]$) and returns $14 - s' + \tilde{s} = 9$ and $\lambda = 5$ as the answer.}
%     \label{fig:ds-runs-fore}
%   \end{figure}

\medskip
\noindent\textbf{Accessing $\col_j(\PBWT)[i]$ from $\subIF$ lists.}
Analogously, for each $1 \le j \le \width$, we store an array $\valF_j$ of length $|\subIF_j|$, where $\valF_j[\tau] = \col_j(\PBWT)[\subIF_j[\tau].b]$.
This requires $O(r_j)$ words per column, and $O(\newR)$ words in total.
Given the index $x$ of the sub-run containing $i$ in $\subIF_j$, we have $\col_j(\PBWT)[i] = \valF_j[x]$, retrievable in $O(1)$ time.
\end{proof}

\section{Applications}
\label{sect-app}

In this section, we present applications that rely on the iterative use of $\Fore/\Back$ queries.
%In Section~\ref{sect-access}, we establish the accessibility of $\newR$.
Section~\ref{sect-app-prefix-search} introduces a PBWT-based solution to prefix searches, while our method to haplotype retrieval is given in Section~\ref{sect-app-extract}.
Due to space constraints, all omitted proofs, pseudocode, and figures in this section are provided in Appendix \ref{app:app}.

Henceforth, set $\rho_j:=|\subIF_j|$ for every $1\le j\le \width$.
According to Theorem \ref{theorem-PBWT-r}, given $i, j$ and $x$, where $1\le j\le \width$, $1\le x\le \rho_j$, and $\subIF_j[x].b \le i \le \subIF_j[x].e$, we define $\ForeAll_j(i, x)$ that returns $(i', x')$ such that $i'=\Fore[i][j]$ and $\subIF_{j+1}[x'].b \le i' \le \subIF_{j+1}[x'].e$.
In other words, $\ForeAll_j(i, x)$ returns $\Fore[i][j]$ and the index of the sub-interval in $\subIF_{j+1}$ that contains $\Fore[i][j]$.

\subsection{A PBWT-Based Solution to Prefix Search}
\label{sect-app-prefix-search}

We first provide a solution under the assumption that all haplotypes $\{S_1,\dots, S_{\height}\}$ are of the same length $\width$ and then generalize it to the case where haplotypes are of arbitrary length.

Similarly as  \emph{sa-interval} in suffix arrays \cite{Manber1993-nn}, given a query pattern $P[1..m]$, we define $\paInterval_j$, with $2\le j\le m$,   as the maximal continuous range of indices in $\col_{j}(\PA)$ such that $P[1..j-1]$ is a prefix of $S_i$ for each $i\in \paInterval_j$.
In particular, $\paInterval_1$ is defined as $[1..\height]$, corresponding to the empty string $P[1..0]$.

\begin{theorem}\label{theorem-prefix-search-mem}
There exists an $O(\newR)$-word data structure constructed over an arbitrarily ordered list $\{S_1, \dots, S_\height\}$ of haplotypes of length $\width$ over the alphabet $\{0,\dots, \sigma-1\}$ such that, given a pattern $P[1..m]$, one can compute in $O(m' \log\log_{\word} \sigma)$ time:
\begin{enumerate}
    \item the longest common prefix $P[1..m']$ between $P[1..m]$ and any haplotype, 
    \item the number $\occ$ of haplotypes prefixed by $P[1..m']$, and 
    \item an index $i$ such that $S_i$ is prefixed by $P[1..m']$.
\end{enumerate}
If $\{S_1, \dots, S_\height\}$ are lexicographically sorted, then the data structure finds, in $O(m' \log\log_{\word} \sigma)$ time, the interval $[\gamma..\gamma']$ such that $S_i$ is prefixed by $P[1..m']$ for every $i \in [\gamma..\gamma']$.
\end{theorem}

In the remainder of this section, we prove Theorem~\ref{theorem-prefix-search-mem}.

\smallskip
\noindent
\textbf{The data structures.}
The construction of the data structure proceeds as follows.
We first build the PBWT matrix $\PBWT$ and the prefix array matrix $\PA$ for the $\height$ haplotypes, each consisting of $\width$ columns and $\height$ rows.
Next, we compute the sub-runs $\subIF_j$ for all $1 \le j \le \width$, as described in Section~\ref{sect-sub-run-back}.
Over these sub-runs, we then construct the $O(\newR)$-space data structure from Theorem~\ref{theorem-PBWT-r} to support $\Fore$ queries efficiently.

We also build the following auxiliary arrays for $1\le j\le \width$. Recall that $\rho_j=|\subIF_j|$.
\begin{itemize}
\item $\sPA_j[1..\rho_j]$: each entry $\sPA_j[\tau]$ ($\tau \in [1..\rho_j]$) stores $\col_j(\PA)[\subIF_j[\tau].b]$;
\item $\sValS_j[1..\rho_j]$: each entry $\sValS_j[\tau]$ stores $\col_j(\PBWT)[\subIF_j[\tau].b]$;
\item $\sCountS_j[1..\rho_j]$: each entry $\sCountS_j[\tau]$ stores the number of occurrences of $\sValS_j[\tau]$ in $\col_j(\PBWT)[1..\subIF_j[\tau].b]$, that is, $\rank_{\sValS_j[\tau]}(\col_j(\PBWT), \subIF_j[\tau].b)$.
\end{itemize}
We then build the data structures for $\rank$ and $\select$ queries over $\sValS_j$ as in Lemma~\ref{lem-rank-select}.
After building these arrays and supporting structures, we discard $\PBWT$ and $\PA$.

The structure of Theorem~\ref{theorem-PBWT-r} occupies $O(\newR)$ words of space.
By Lemma~\ref{lem-size-sub-runs}, we have $\sum_j \rho_j = O(\sum_j r_j)=O(\newR)$.
Hence, storing all the arrays and auxiliary $\rank/\select$ structures requires $O(\sum_j \rho_j) = O(\newR)$ words in total.
Consequently, the overall space complexity of our data structure is $O(\newR)$ words.

% \medskip
% \noindent
% \textbf{Forward Queries.}
\medskip
\noindent
\textbf{The query algorithm.}
The algorithm \texttt{Partial Prefix Search($P[1..m]$)} proceeds iteratively by updating  a quintuple $(b, e, x, x', index)$, where 
\begin{itemize}
    \item $b$ and $e$, drawn from $\{1,\dots, \height\}$, store some row indices,
    \item $x$ and $x'$, drawn from $\{1,\dots, |\subIF_j|\}$ for some column $j$, always store indices that satisfy $b\in \subIF_j[x]$ and $e\in \subIF_j[x']$, respectively, and 
    \item $index$, drawn from $\{1,\dots, \height\}$, always stores $\col_j(\PA)[b]$.
\end{itemize}
% Initially, we set the quintuple such that
% \begin{itemize}
%     \item $[b, e]$ is $[1, h]$, which is $\paInterval_1$,
%     \item $x=1$ and $x'=|\subIF_1|$, and
%     \item $index:=\sPA_1[1]$.
% \end{itemize}
Initially, we set the quintuple as
$[b,e]=[1,h]=\paInterval_1$, $x=1$, $x'=|\subIF_1|$, and $index=\sPA_1[1]$.

In the \( j \)-th iteration (\( 1 \le j \le m \)), the following steps are performed:
\begin{itemize}
    \item We first determine the positions \(\tilde{b}\) and \(\tilde{e}\), 
    where \( b \le \tilde{b} \le \tilde{e} \le e \) denote the indices of the first and last occurrences of \( P[j] \) in \( \col_j(\PBWT) \) within the range \([b, e]\). 
    Formally, \(\tilde{b} = \min\{ \ell \mid b \le \ell \le e,\, \col_j(\PBWT)[\ell] = P[j] \}\) and 
    \(\tilde{e} = \max\{ \ell \mid b \le \ell \le e,\, \col_j(\PBWT)[\ell] = P[j] \}\). 
    If neither \(\tilde{b}\) nor \(\tilde{e}\) exists, the interval \(\paInterval_{j+1}\) 
    corresponding to the prefix \( P[1..j] \) is empty.  
    Otherwise, we locate the indices \(\tilde{x}\) and \(\tilde{x}'\) such that 
    \(\tilde{b} \in \subIF_j[\tilde{x}]\) and \(\tilde{e} \in \subIF_j[\tilde{x}']\), respectively.  
    In the pseudocode, we show that all of \(\tilde{b}\), \(\tilde{e}\), \(\tilde{x}\), and 
    \(\tilde{x}'\) can be obtained using \(\rank\) and \(\select\) queries 
    over the array \(\sValS_j\).  
    In addition, if \(\tilde{b} \ne b\), we update the variable 
    \(\textit{index} := \sPA_j[\tilde{x}]\).  
    We will later prove that \(\textit{index}\) always stores  \(\col_j(\PA)[b]\).
    \item Second, if \( j \le m \) and \(\paInterval_{j+1} \neq \emptyset\), 
    we apply the queries \(\ForeAll_j(\tilde{b}, \tilde{x})\) 
    and \(\ForeAll_j(\tilde{e}, \tilde{x}')\) according to 
    Theorem~\ref{theorem-PBWT-r}, and update 
    \((b, x) := \ForeAll_j(\tilde{b}, \tilde{x})\) and 
    \((e, x') := \ForeAll_j(\tilde{e}, \tilde{x}')\), respectively.  
    Consequently, for \( j < m \), it follows that \( b = \paInterval_{j+1}.b \) and 
    \( e = \paInterval_{j+1}.e \).  
    Thus, \(\paInterval_{j+1}\), corresponding to the prefix \( P[1..j] \), 
    is obtained and stored as \([b, e]\), with the values \(x\) and \(x'\) 
    updated accordingly.  
    If \( j = m \), we instead set \((b, x) := (\tilde{b}, \tilde{x})\) and 
    \((e, x') := (\tilde{e}, \tilde{x}')\).

    \item Third, we increment \( j \) by 1 and proceed to the next iteration.
\end{itemize}

The $j$-th iteration ends when   $\paInterval_{j+1}=\emptyset$ or when $j = m+1$.  
The following observation is crucial: in either case, the longest common prefix shared by $P[1..m]$ and any haplotype is $P[1..j-1]$.

If $\paInterval_{j+1}=\emptyset$,  then $[b, e]$ is $\paInterval_{j}$, corresponding to the prefix $P[1..j-1]$.
If $j=1$, then the longest common prefix shared by $P[1..m]$ and any haplotype is an empty string. 
Otherwise, the longest common prefix is $P[1..j-1]$ by the observation mentioned earlier, and the number of haplotypes prefixed by $P[1..j-1]$ is $e-b+1$ by the definition of $\paInterval_{j}$; therefore, we return $P[1..j-1], e-b+1$, and $index$ as the answer.

If instead $j = m+1$, it follows that $P[1..m]$ is the longest common prefix and that $b$ and $e$ store, respectively, the indices of the and last occurrences of $P[m]$ in $\col_m(\PBWT)[\hat{b}, \hat{e}]$, where $\hat{b}=\paInterval_m.b$ and $\hat{e}=\paInterval_m.e$.
In this case, we compute the numbers of occurrences of $P[m]$ in $\col_m(\PBWT)[1..b]$ and $\col_m(\PBWT)[1..e]$, which are $\sCountS_{m}[x] + b - \subIF_m[x].b$ and $\sCountS_{j}[x'] + e - \subIF_j[x'].b$, respectively, and store them in variables $count_1$ and $count_2$.
The number of haplotypes prefixed by $P[1..m]$ is $count_2-count_1+1$.
In the end, we return $P[1..m], count_2-count_1+1, index$ as the answer; recall that $index$ stores $\col_m(\PA)[b]$.
The pseudocode of the algorithm can be found in Appendix \ref{app-pseudo-prefix-search}.

We have shown that the procedure \texttt{Partial Prefix Search($P[1..m]$)} identifies the longest prefix $P[1..m']$ shared between $P[1..m]$ and any haplotype, and counts the number of haplotypes prefixed by $P[1..m']$. 
It remains to show that the variable \textit{index} returned by the procedure indeed stores the index of one of these haplotypes. 
To this end, it suffices to prove that the invariant $index = \col_j(\PA)[b]$ holds in each iteration $j$.

\begin{lemma}\label{lem-index-store}
    It follows that $index=\col_j(\PA)[b]$ in each iteration $j$ with $1\le j\le m$.
\end{lemma}

The proof of this lemma is deferred to Appendix \ref{app-index-store}, due to space limitations.

Moreover, by Proposition~\ref{prop-lf-map}(a), it follows that the variable $index$ returned by the algorithm stores the smallest index $\ell$ such that $S_{\ell}$ is prefixed by $P[1..m']$; that is, $index = \min \{\, \ell \mid 1 \le \ell \le \height \text{ and } P[1..m'] \text{ is a prefix of } S_{\ell} \,\}$.

\medskip
\noindent
\textbf{The running time analysis.} The while-loop in the algorithm performs at most $m$ iterations.
Recall that $m'$ denotes the length of the longest common prefix to return.
So, $\paInterval_{m'+2}$ does not exist, and at most $\min(m'+1, m)$ iterations are executed.
In each iteration $j$, the queries $\rank$ and $\select$ over the array $\sValS_j$, as well as the query $\ForeAll$ are called at most twice, respectively.
By Lemma \ref{lem-rank-select} and Theorem \ref{theorem-PBWT-r}, a $\rank$ query takes $O(\log \log_{\word} \sigma)$ time, a query $\select$ or $\ForeAll_j$ takes $O(1)$ time.
Hence, the overall running time of the algorithm is bounded by $O(m'\log \log_{\word} \sigma)$.

\medskip
\noindent
\textbf{When $\{S_1, S_2, \dots, S_{\height}\}$ are sorted lexicographically.}  
In this case, we first invoke the procedure \texttt{Partial Prefix Search($P[1..m]$)} to obtain $P[1..m']$, the count $\occ$, and the index $i$ of a haplotype that is prefixed by $P[1..m']$.  
We then compute the interval $[\gamma..\gamma']$ as $[i..i+\occ-1]$.
To verify correctness, observe that if $S_1, S_2, \dots, S_{\height}$ are presorted lexicographically, then the list $\{\gamma, \gamma+1, \dots, \gamma'-1, \gamma'\}$ forms a sub-list of $\col_j(\PA)$ for every $1 \le j \le \min(m'+1, m)$.  
Recall that $i = \min \{\, \ell \mid 1 \le \ell \le \height \text{ and } P[1..m'] \text{ is a prefix of } S_{\ell} \,\}$.
Hence, it follows that $i = \gamma$ and $\gamma' = i + \occ - 1$.

This completes the proof of Theorem~\ref{theorem-prefix-search-mem}.  
In Corollary~\ref{coro-sorted}, we extend this result to enumerate all haplotype indices prefixed by \( P[1..m'] \); the proof appears in Appendix~\ref{app-coro-sorted}.

\begin{corollary}\label{coro-sorted}
There is an $O(\newR + \height)$-word data structure built over an arbitrary ordered list $\{S_1, \dots, S_\height\}$ of haplotypes such that, given a pattern $P[1..m]$, %the following hold:
the indices of the haplotypes prefixed by $P[1..m']$ can be enumerated in $O(m' \log\log_{\word} \sigma + \occ)$ time, where $P[1..m']$ is the longest common prefix between $P[1..m]$ and any haplotype and $\occ$ denotes the number of haplotypes prefixed by $P[1..m']$.
% \begin{enumerate}
%     \item The longest common prefix $P[1..m']$ between $P[1..m]$ and any haplotype, as well as the number $\occ$ of haplotypes prefixed by $P[1..m']$, can be computed in $O(m' \log\log_{\word} \sigma)$ time.
%     \item The indices of the haplotypes prefixed by $P[1..m']$ can be enumerated in $O(m' \log\log_{\word} \sigma + \occ)$ time.
% \end{enumerate}
\end{corollary}

We extend our solution to haplotypes of arbitrary lengths in Appendix~\ref{app-arbitrary-length}.

\subsection{Haplotype Retrieval within $O(\newR)$ Space}
\label{sect-app-extract}

\begin{theorem}\label{theorem-PBWT-extract}
The list of haplotypes $\{S_1, \ldots, S_\height\}$, each of length $\width$, 
can be represented in $O(\newR)$ words of space, 
allowing $S_i$ with any $1 \le i \le \height$ 
to be retrieved in $O(\width+\log\log_{\word} \height)$ time.
%If the alphabet size is bounded by $\word^{O(1)}$, then the query time is bounded by $O(\width)$.
\end{theorem}

\begin{proof}
We construct the data structure, denoted by $DS$, as described in Theorem~\ref{theorem-PBWT-r} for supporting $\Fore$ queries.  
Let $p_{\tau}$, for $1\le \tau \le |\subIF_1|$, be the starting positions of the $\tau$-th sub-run in $\col_1(\PBWT)$.  
Clearly, these positions $p_1, p_2, \dots$ are sorted in increasing order. 
By Lemma \ref{lem-size-sub-runs}, the list consists of at most $2\newR$ positions.
We build the data structure from Lemma~\ref{lem-pred} over the list $p_1, p_2, \dots$ to support predecessor queries.  
Since these positions are drawn from the universe $\{1, \dots, \height\}$, each predecessor query can be answered in $O(\log \log_{\word} \height)$ time.  
The data structure requires $O(2\newR+\newR)=O(\newR)$ words of space.

Given a query position $i$, we apply Lemma~\ref{lem-pred} to find the index $x_1$ of the sub-run that contains $i$; that is, $x_1$ is the position of the predecessor of $i$ in the list $p_1, p_2, \dots$.  
Using $x_1$, we can retrieve $\col_1(\PBWT)[i]$ and $\Fore[i][1]$ in constant time via $DS$.  
The former gives the entry $S_i[1]$, while the latter provides both the position $i_2$---where $i$ is stored in $\col_2(\PA)$---and the index $x_2$ of the sub-run in $\col_2(\PBWT)$ that contains $i_2$.  
With $x_2$ and $i_2$, we can obtain the entry $S_i[2]$ and continue in the same manner for the subsequent columns.  
The algorithm terminates after all $\width$ entries of $S_i$ have been retrieved.

The predecessor query is invoked once, costing $O(\log \log_{\word} \height)$ time by Lemma \ref{lem-pred}.  
The $\Fore$ query is invoked $\width$ times, and exactly $\width$ entries of the matrix $\PBWT$ are accessed.  
Thus, the total query time is bounded by $O(\log \log_{\word} \height + \width)$.
\end{proof}

% \begin{corollary}\label{coro-extraction}
% If the size of the alphabet from which the haplotypes are drawn is bounded by $\word^{O(1)}$ and all haplotypes are distinct, then the query time for extraction is $O(\width)$.
% \end{corollary}

% \begin{proof}
%   Since all haplotypes are distinct, it follows that $\sigma^\width \ge \height$, which implies $\width \ge \log_{\sigma} \height$.  
% If $\sigma$ is bounded by $\word^{O(1)}$, then $\width = \Omega(\log_{\word} \height)$.  
% Therefore, the overall running time is $O(\log \log_{\word} \height + \width) = O(\width)$.  
% \end{proof}

\section{Conclusions}
\label{sect-conclusion}

In this work, we presented new $O(\newR)$-space data structures and algorithms that support efficient forward and backward stepping operations on the PBWT. 
Our data structure enables constant-time computation of both $\Fore$ and $\Back$ operations. 
We also established lower and upper bounds on $\newR$. 
To demonstrate the utility of the optimal-time mapping, we revisited two applications: prefix search and haplotype retrieval. 
For the former, we proposed an $O(\newR + \height)$--word data structure that answers each query in $O(m' \log\log_{\word} \sigma)$ time. 
When the haplotypes are provided in lexicographic order, the space requirement reduces to $O(\newR)$ words. 
Moreover, our PBWT-based approach to prefix search naturally extends to haplotypes of arbitrary length. 
For the latter, we designed an $O(\newR)$-word data structure that supports haplotype retrieval in $O(\width + \log\log_{\word} \height)$ time. 
While this work primarily  focuses on the theoretical approach for achieving constant-time mapping in the run-length encoded PBWT, it also opens new directions for practical implementations and other applications,  such as accelerating the computation of SMEMs \cite{Cozzi2023-iv} or MPSC in haplotype threading \cite{sanaullah2022haplotype, Bonizzoni2024-dj}. We leave these aspects for future work.
%\textcolor{red}{Future work includes implementing our data structures and algorithms to experimentally evaluate their practical performance. In particular, we plan to compare our approach with $\mu$-PBWT in terms of query time and space cost, providing empirical evidence of efficiency and scalability.}

%\pb{I am not convinced about the last two applications since a natural question would be why we did not consider these two in the paper. Moreover, shall we think that we can accelerate the computation of matching statistics.. or MPSC?...}
%\pb{The applications discussed in this paper show the potentiality of our data structure in  supporting PBWT-queries: as a future work we plan to extend it to the Minimal Positional Substring cover \cite{sanaullah2022haplotype, Bonizzoni2024-dj} using matching statistics revisited under the move data structure?.}
%Future work includes extending our framework to support additional PBWT-based queries such as matching statistics and Minimal Positional Substring Covers—under the same $O(\newR)$-space constraint. 
%It would also be interesting to further explore the relationship between $\newR$ and $r$. 
%Establishing an upper bound on $\newR$ in terms of $r$ could bring us closer to resolving the long-standing question of whether $r$ itself is accessible.

%%
%% Bibliography
%%

%% Please use bibtex, 

\bibliography{move-pbwt-ref}

\newpage
\appendix

\section{The Figure Omitted in Section \ref{sect-intro}}
\label{app-pbwt-pa-exp}
\begin{figure}[H]
  \centering
  \begin{subfigure}[b]{0.3\textwidth}
    \centering
    \includegraphics[width=\textwidth]{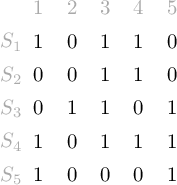}
    \caption{The Input haplotypes}
    \label{fig:hap-exp}
  \end{subfigure}
  \hfill
  \begin{subfigure}[b]{0.3\textwidth}
    \centering
    \includegraphics[width=\textwidth]{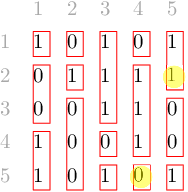}
    \caption{The PBWT with runs}
    \label{fig:pbwt-exp}
  \end{subfigure}
  \hfill
  \begin{subfigure}[b]{0.3\textwidth}
    \centering
    \includegraphics[width=\textwidth]{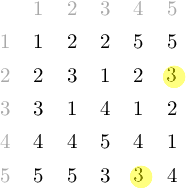}
    \caption{PA}
    \label{fig:pa-exp}
  \end{subfigure}
  \caption{Example of $\PBWT$ and $\PA$ built for bi-allelic haplotypes $\{S_1, S_2, S_3, S_4, S_5\}$. The operation $\Fore[5][4]$ returns $2$, and the position 2 is in the first run of the fifth column of the PBWT.}
  \label{fig:pbwt-pa-exp}
\end{figure}

\section{The Proofs Omitted in Section \ref{sect-prelim}}

\subsection{The Proof of Proposition \ref{prop-lf-map}}
\label{app-prop-lf-map}

\begin{proof}
When constructing the $j$-th column of $\PBWT$, note that the prefixes of the haplotypes up to column $j-1$ are stably sorted.
Therefore, statements (a) and (b) follow directly from this stable ordering.
\end{proof}

% \section{The Proofs Omitted in Section \ref{sect-bounds-newR}}
% \label{app-bounds-newR}

% \subsection{The Proof of Corollary \ref{prop-height-r}}

% \begin{proof}
%     The statement holds because the number of distinct haplotypes in $\{S_1, S_2, \dots, S_{\height}\}$ is at most $\height'' + 1$.
% \end{proof}

% \subsection{The Proof of Theorem \ref{theorem-upper-newR}}

% \begin{proof}
% By Lemmas~\ref{lemma-r-ell} and~\ref{lem-bound-ell_j}, we have $r_j \le \height'' + 1$ for all $j$.
% Therefore, $\newR = \sum_{1 \le j \le \width} r_j 
%     \le \sum_{1 \le j \le \width} (\height'' + 1)
%     = \width (\height'' + 1).$
% \end{proof}

\section{The Figure Omitted from Section \ref{sect-three-overlap-normal}}
\label{app:normal-alg}
\begin{figure}[H]
		\centering
		%\captionsetup{justification=centering}
		\includegraphics[scale=1]{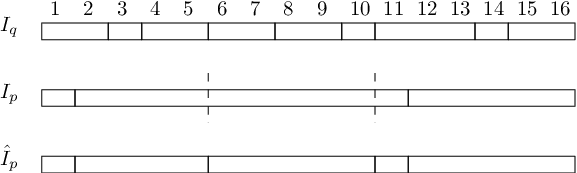}
		\caption{
Example of the normalization algorithm. 
Given $I_p = \{[1,1], [2,11], [12,16]\}$ and 
$I_q = \{[1,2], [3,3], [4,5], [6,7], [8,9], [10,10], [11,13], [14,14], [15,16]\}$, 
the algorithm outputs $\hat{I}_p = \{[1,1], [2,5], [6,10], [11,11],  [12,16]\}$, 
where each interval in $\hat{I}_p$ overlaps at most three intervals in $I_q$.
}
    \label{fig:normal-alg}
\end{figure}

\section{The Figure Omitted from Section \ref{sect-sub-run-back}}
\label{app-fore}
\begin{figure}[H]
  \centering
  \begin{subfigure}[b]{0.44\textwidth}
    \centering
    \includegraphics[width=\textwidth]{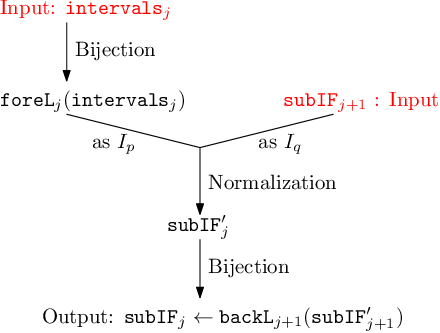}
    \caption{The scheme of building $\subIF_j$}
  \end{subfigure}
  \hfill
  \begin{subfigure}[b]{0.54\textwidth}
    \centering
    \includegraphics[width=\textwidth]{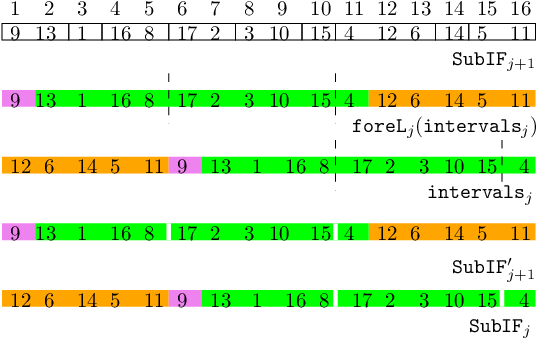}
    \caption{An Example of building $\subIF_j$}
  \end{subfigure}
  \caption{\label{fig:fore} Illustration of the algorithm scheme for building sub-runs in $\subIF_j$ and an example. In this example, $\subIF_{j+1} = \{\,[1,2], [3,3], [4,5], [6,7], [8,9], [10,10], [11,13], [14,14], [15,16]\,\}$ and 
$\inter_j = \{\,[1,5], [6,6], [7,16]\,\}$.  
The bijective function $\BiFore_{j}$ maps the list $\inter_{j}$ to the list 
$\{\,[1,1], [2,11], [12, 16]\,\}$.
Intervals highlighted in the same color contain the same haplotype indices and indicate corresponding pairs under this bijection.  
After applying the normalization algorithm to $\BiFore_j(\inter_j)$ and $\subIF_{j+1}$, the intervals in $\BiFore_j(\inter_j)$ are partitioned into  
$\subIF'_j = \{\,[1,1], [2,5], [6,10], [11,11], [12,16]\,\}$.  
Each interval in $\subIF'_{j+1}$ overlaps with at most three intervals in $\subIF_{j+1}$.
Finally, by setting \(\subIF_j = \BiBack_{j+1}(\subIF'_{j+1})\), the intervals \(\inter_j\) are partitioned into \(\subIF_j = \{[1,5], [6,6], [7,10], [11,15], [16,16]\}\).}
\end{figure}

\section{The Figure Omitted from Section \ref{sect-ds-fore-back}}
\label{app-ds-runs-fore}
\begin{figure}[H]
    \centering
    \includegraphics[scale=1]{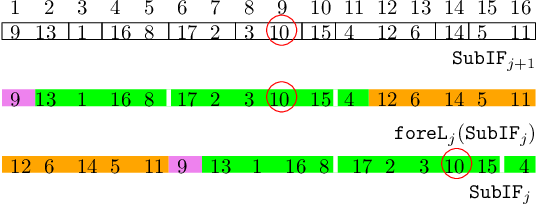}
\caption{Example illustrating the data structure and algorithm for $\Fore$ queries. 
The fourth interval $[11,15]$ in $\subIF_j$ corresponds to the third interval $[6,10]$ in $\BiFore_j(\subIF_j)$, which overlaps three intervals in $\subIF_{j+1}$: the fourth $[6,7]$, the fifth $[8,9]$, and the sixth $[10,10]$. 
Accordingly, the data structure $F^4_j$ built for the interval $[11,15]$ stores the quintuples $\{(11,6,6,7,4), (11,6,8,9,5), (11,6,10,10,6)\}$ in the form $(s', \tilde{s}, s, t, \lambda)$. 
Given a query $\Fore[14][j]$ with index $4$, the algorithm finds the tuple $(s'=11, \tilde{s}=6, s=8, t=9, \lambda=5)$ in $F^4_j$ (since $14 - s' + \tilde{s} = 14 - 11 + 6 = 9 \in [8,9]$) and returns $14 - s' + \tilde{s} = 9$ and $\lambda = 5$ as the answer.}
    \label{fig:ds-runs-fore}
  \end{figure}
  \newpage

\section{The Details Omitted in Section \ref{sect-app}}
\label{app:app}

\subsection{The Pseudocode for Prefix Searches}
\label{app-pseudo-prefix-search}

%\begin{figure*}[!h]
\texttt{Partial Prefix Search($P[1..m]$)} \\
01.\hspace*{1cm} $(b, e, x, x', index) \leftarrow (1, \height, 1, |\subIF_1|, \sPA_1[1])$; \\
02.\hspace*{1cm} $j \leftarrow 1$; \\
03.\hspace*{1cm} \textbf{while} $j \le m$ \textbf{do} \\
04.\hspace*{2cm} $c \leftarrow P[j]$; \\
05.\hspace*{2cm} $(\tilde{b}, \tilde{x}) \leftarrow (b, x)$; \\
06.\hspace*{2cm} \textbf{if} $\sValS_j[x] \ne c$ \textbf{then} \\
07.\hspace*{3cm} $count \leftarrow \rank_c(\sValS_j, x)$; \\
08.\hspace*{3cm} $\tilde{x} \leftarrow \select_c(\sValS_j, count+1)$; \\
09.\hspace*{3cm} \textbf{if} $\tilde{x} = |\subIF_j|+1$ \textbf{then} \\
10.\hspace*{4cm} \textbf{break}; \\
11.\hspace*{3cm} $\tilde{b} \leftarrow \subIF_j[\tilde{x}].b$; \\
12.\hspace*{3cm} $index \leftarrow \sPA_j[\tilde{x}]$; \\
13.\hspace*{2cm} $(\tilde{e}, \tilde{x}') \leftarrow (e, x')$; \\
14.\hspace*{2cm} \textbf{if} $\sValS_j[x'] \ne c$ \textbf{then} \\
15.\hspace*{3cm} $count \leftarrow \rank_c(\sValS_j, x')$; \\
16.\hspace*{3cm} $\tilde{x}' \leftarrow \select_c(\sValS_j, count)$; \\
17.\hspace*{3cm} $\tilde{e} \leftarrow \subIF_j[x'].e$; \\
18.\hspace*{2cm} \textbf{if} $j < m$ \textbf{then} \\
19.\hspace*{3cm} $(b, x) \leftarrow \ForeAll_j(\tilde{b}, \tilde{x})$; \\
20.\hspace*{3cm} $(e, x') \leftarrow \ForeAll_j(\tilde{e}, \tilde{x}')$; \\
21.\hspace*{2cm} \textbf{else}  \\
22.\hspace*{3cm} $(b, x) \leftarrow (\tilde{b}, \tilde{x})$; \\
23.\hspace*{3cm} $(e, x') \leftarrow (\tilde{e}, \tilde{x}')$; \\
24.\hspace*{2cm} $j \leftarrow j + 1$; \\
25.\hspace*{1cm} \textbf{if} $j = 1$ \textbf{then} \\
26.\hspace*{2cm} \textbf{return} the longest common prefix is empty; \\
27.\hspace*{1cm} \textbf{else if} $1<j\le m$ \\
28.\hspace*{2cm} \textbf{return} $P[1..j-1], e-b+1, index$; \\
29.\hspace*{1cm} \textbf{else}  \\
30.\hspace*{2cm} $count_1 \leftarrow \sCountS_{m}[x] + b - \subIF_m[x].b$; \\
31.\hspace*{2cm} $count_2 \leftarrow \sCountS_{m}[x'] + e - \subIF_m[x'].b$; \\
32.\hspace*{2cm} $count \leftarrow count_2 - count_1 + 1$; \\
33.\hspace*{2cm} \textbf{return} $P[1..m], count, index$; 
%\end{figure*}

\subsection{The Proof of Lemma \ref{lem-index-store}}
\label{app-index-store}

\begin{proof}
    We give a proof by induction.
    When $j=1$, the procedure sets $b:=1$ and $index:=\sPA_1[1]$.
    Since $\col_j(\PA)[b]=\col_j(\PA)[1]=\sPA_1[1]$, we have $index=\col_j(\PA)[b]$; the base case for $j=1$ holds trivially. 

    Assume by induction that at the beginning of the $j$-th iteration, we have $index=\col_{j}(\PA)[b]$.
    Consider the $j$-th iteration.
    Note that in the beginning of this iteration, we have $\sValS_j[x]=\col_j(\PBWT)[b]$, as $b\in \subIF_j[x]$.
    
    If $\sValS_j[x]=P[j]$, then $\col_j(\PBWT)[b]=P[j]$; in this case, the variable $index$ remains in this iteration, and the variable $b$ is set to $\Fore[b][j]$.
    Note that by the definition of $\Fore$ queries, we have $\col_{j}(\PA)[b]=\col_{j+1}(\PA)[\Fore[b][j]]$.
    By the inductive assumption, we have $index=\col_{j}(\PA)[b]$, thereby $index=\col_{j+1}(\PA)[\Fore[b][j]]$; therefore, at the beginning of the next iteration, the statement still holds.

    Otherwise, we have $\col_j(\PBWT)[b]\ne P[j]$; in this case, the procedure finds the index $\tilde{b}$ of the first occurrence of $P[j]$ in $\col_j(\PBWT)$ within the range $[b..e]$ and the index $\tilde{x}$ that satisfies $\tilde{b}= \subIF_j[\tilde{x}].b$, and sets $index:=\sPA_j[\tilde{x}]$.
    Since $\col_{j}(\PA)[\tilde{b}]=\sPA_j[\tilde{x}]$, we have $\col_{j}(\PA)[\tilde{b}]=index$.
    If $j=m$, then the procedure sets $b:=\tilde{b}$; therefore, we have $\col_m(\PA)[b]=\col_m(\PA)[\tilde{b}]=\sPA_m[\tilde{x}]=index$.
    When the while-loop terminates in the $m$-th iteration, we have $\col_m(\PA)[b]=index$.
    If $j<m$, then $b$ is set to $\Fore[\tilde{b}][j]$.
    By the definition of $\Fore$ queries, we have $\col_{j}(\PA)[\tilde{b}]=\col_{j+1}(\PA)[\Fore[\tilde{b}][j]]$.
    Recall that $\col_{j}(\PA)[\tilde{b}]=index$, so it follows that $index=\col_{j+1}(\PA)[\Fore[\tilde{b}][j]]=\col_{j+1}(\PA)[b]$.
    Thus, at the beginning of the next iteration, the statement  holds, completing the proof.
\end{proof}

\subsection{The Proof of Corollary \ref{coro-sorted}}
\label{app-coro-sorted}

\begin{proof}
We first sort all haplotypes in lexicographic order, build the data structure of Theorem~\ref{theorem-prefix-search-mem} over the sorted list, and store the permutation $\pi^{-1}(i)$ for $1 \le i \le \height$ such that $S_{\pi(1)} \prec S_{\pi(2)} \prec \dots \prec S_{\pi(\height)}$.
The resulting data structure uses $O(\newR+\height)$ words of space.

Given a query pattern $P[1..m]$, Theorem~\ref{theorem-prefix-search-mem} allows us to find the longest common prefix $P[1..m']$ in $O(m' \log\log_{\word} \sigma)$ time, as well as the interval $[\gamma, \gamma']$ such that $S_{\pi(i)}$ for any $\pi(i) \in [\gamma, \gamma']$ is prefixed by $P[1..m']$.  
The number $\occ$ of occurrences is then $\gamma' - \gamma + 1$, and the original indices, before sorting, of these haplotypes prefixed by $P[1..m']$ are $\{\pi^{-1}(\tau) \mid \tau \in [\gamma, \gamma']\}$.
Using the permutation $\pi^{-1}$ together with the interval $[\gamma, \gamma']$, the list of indices can be reported in $O(\occ)$ time.
\end{proof}

\subsection{Handling Haplotypes of Arbitrary Lengths}
\label{app-arbitrary-length}

We consider a general setting where the haplotypes $\{S_1, \dots, S_\height\}$ have arbitrary length.
We append a special symbol $\texttt{\#} \notin \{0, \dots, \sigma-1\}$---assumed to be the smallest---to the end of each haplotype.
Next, we construct the PBWT representation $\PBWT$ and the prefix arrays $\PA$ for the new haplotypes, which consist of $\height + \sum_{i \in [1..\height]} |S_i|$ symbols in total.

The construction of $\PBWT$ proceeds column by column from left to right, as in the case of haplotypes of equal length. 
In this setting, however, both $\PBWT$ and the $\PA$ are no longer matrices; instead, each consists of $\width$ columns of possibly different lengths, denoted by $\{\PBWT_1, \PBWT_2, \dots, \PBWT_{\width}\}$ and $\{\PA_1, \PA_2, \dots, \PA_{\width}\}$, respectively. 
Here, $\width$ denotes the maximum extended haplotype length among $\{S_1, \dots, S_\height\}$, that is, 
$\width = \max \{ |S_i| +1 \mid 1 \le i \le \height \}$. 
When constructing the entries of $\PBWT_j$ and $\PA_j$, we exclude any haplotype $S_i$ such that $j>|S_i|+1$.
Figure \ref{fig:pbwt-pa} illustrates an example of $\PBWT$ and $\PA$ built for haplotypes of arbitrary length.

\begin{figure}[!h]
  \centering
  \begin{subfigure}[b]{0.3\textwidth}
    \centering
    \includegraphics[width=\textwidth]{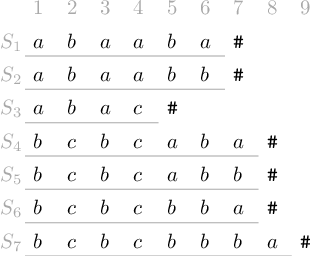}
    \caption{Extended haplotypes}
    \label{fig:e-haplotypes}
  \end{subfigure}
  \hfill
  \begin{subfigure}[b]{0.3\textwidth}
    \centering
    \includegraphics[width=\textwidth]{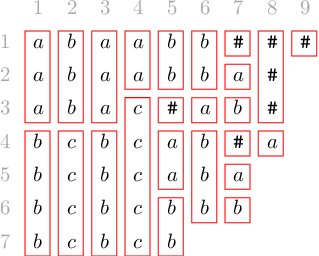}
    \caption{$\PBWT$}
    \label{fig:pbwt}
  \end{subfigure}
  \hfill
  \begin{subfigure}[b]{0.3\textwidth}
    \centering
    \includegraphics[width=\textwidth]{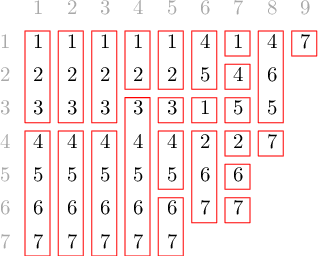}
    \caption{$\PA$}
    \label{fig:pa}
  \end{subfigure}
  \caption{Example of $\PBWT$ and $\PA$ built for haplotypes of arbitrary length.}
  \label{fig:pbwt-pa}
\end{figure}

The construction procedure described above ensures that each column $\PBWT_j$ stores a permutation of all entries $S_i[j]$ for every $1 \le i \le \height$ and $j \le |S_i|+1$.
Note that $S_i[|S_i|+1]=\texttt{\#}$.
As a result, the symbol $\texttt{\#}$ appears in $\PBWT$ exactly $\height$ times across all columns.
Let $\newR$ denote the total number of runs in $\PBWT$, and let $\newR'$ denote the number of runs excluding those composed entirely of $\texttt{\#}$.
Since $\texttt{\#}$ occurs $\height$ times in $\PBWT$, it follows that $\newR \le \newR' + \height$.

Upon the runs of $\PBWT$ for the new haplotypes trailing with $\texttt{\#}$, we build the same data structures as we have seen before for constructing sub-runs $\subIF_j$'s, for $\Fore$ queries, and for prefix searches (Corollary \ref{coro-sorted}).
The space cost is bounded by $O(\newR+\height)$ words, and a prefix search query can be answered in $O(m'\log \log_{\word} (\sigma+1)+\occ)$ time, since the new alphabet set is $\{0, \dots, \sigma-1\}\cup \{\texttt{\#}\}$.
Note that a prefix search query never calls $\Fore[i][j]$ for any $i$ and $j$ such that $\PBWT_j[i]=\texttt{\#}$.

\end{document}